\pdfoutput=1
\documentclass[nofootinbib,aps]{revtex4}
\usepackage{amssymb,amsmath,graphicx,color,microtype,amsthm}
\usepackage{enumerate}
\usepackage{slashed}
\usepackage{xcolor}
\usepackage{placeins}
\usepackage{hyperref}
\hypersetup{
	colorlinks=true,
	linkcolor=blue,
	filecolor=blue,      
	urlcolor=blue,
	pdftitle={Probing P and CP Violations on the Cosmological Collider},
	bookmarks=true,
	citecolor=blue
}

\newtheorem*{theorem}{Theorem}
\usepackage[paperwidth=210mm,paperheight=297mm,centering,hmargin=2cm,vmargin=2.5cm]{geometry}
\graphicspath{{fig/}}

\def\bge{\begin{equation}}
\def\ede{\end{equation}}
\def\bga{\begin{aligned}}
\def\eda{\end{aligned}}
\def\bgp{\begin{pmatrix}}
\def\edp{\end{pmatrix}}
\def\bgs{\begin{subequations}}
\def\eds{\end{subequations}}
\newcommand{\order}[1]{\mathcal{O}({#1})}
\def\di{{\mathrm{d}}}

\def\la{\langle}\def\ra{\rangle}

\def\to{\rightarrow}

\def\ii{\mathrm{i}}

\def\ep{\epsilon}

\def\si{\sigma}

\def\Re{\mathrm{Re}\,}
\def\Im{\mathrm{Im}\,}

\newcommand{\wt}[1]{\mkern 2mu \widetilde{\mkern -2mu #1 \mkern -2mu}\mkern 2mu}

\begin{document}
\title{Probing P and CP Violations on the Cosmological Collider}
\author{Tao Liu$^{1,2}$}
\email{taoliu@ust.hk}
\author{Xi Tong$^{1,2}$}
\email{xtongac@connect.ust.hk}
\author{Yi Wang$^{1,2}$}
\email{phyw@ust.hk}
\author{Zhong-Zhi Xianyu$^{3}$}
\email{zxianyu@g.harvard.edu}
\affiliation{${}^1$Department of Physics, The Hong Kong University of Science and Technology, \\
	Clear Water Bay, Kowloon, Hong Kong, P.R.China}
\affiliation{${}^2$The HKUST Jockey Club Institute for Advanced Study, The Hong Kong University of Science and Technology, \\
	Clear Water Bay, Kowloon, Hong Kong, P.R.China}
\affiliation{${}^3$Department of Physics, Harvard University, 17 Oxford Street, Cambridge, MA 02138, USA}

\begin{abstract}
	In direct analogy to the 4-body decay of a heavy scalar particle, the 4-point correlation function of primordial fluctuations carries P and CP information. The CP violation appears as a P-odd angular dependence in the imaginary part of the trispectrum in momentum space. We construct a model with axion-like couplings which leads to observably large CP-violating trispectrum for future surveys. Furthermore, we show the importance of on-shell particle production in observing P- and CP-violating signals. It is impossible to observe these signals from local 4-scalar EFT operators that respect dilation symmetry, and thus any such observation can rule out single-field EFT with sufficiently small slow-roll parameters. This calculation opens a new frontier of studying P and CP at very high energy scales.  
\end{abstract}

\maketitle
\section{Introduction}

The study of discrete symmetries and their breaking has a rich history and remains an active field of research in fundamental physics. Parity (P), charge conjugation (C) and their joint transformation, namely CP, were all known to be violated in weak interaction \cite{Lee:1956qn,Christenson:1964fg}. On the other hand, the combination of CP with time reversal T, namely CPT, is known to be an unbroken symmetry of any quantum field theory respecting Lorentz symmetry, which is the celebrated CPT theorem \cite{Schwinger:1951xk,Bell:1996nh}. 

In the Standard Model (SM) of particle physics, the weak interaction maximally breaks P but preserves CP for one generation of fermions. The CP violation in the electroweak sector appears only when there are at least three generations of fermions, for there is a CP-violating complex phase in the Cabbibo-Kobayashi-Maskawa matrix \cite{Kobayashi:1973fv} that cannot be rotated away. 
The $\theta$ terms of the gauge fields introduce additional CP-violations in the theory. The $\theta$ terms are not observable in the electroweak dynamics, while introducing a CP-violating neutron electric dipole moment (EDM) in the strong sector. The long-undetected neutron EDM puts a tight constraint on the strong vacuum angle $\theta_c<10^{-10}$ which is unnaturally small \cite{Baker:2006ts}. This so-called strong CP problem can be solved by the Peccei-Quinn mechanism \cite{Peccei:1977hh}, in which $\theta_c$ is promoted to a dynamical field known as axion \cite{Weinberg:1977ma,Wilczek:1977pj}, and stabilized to a zero value by the QCD instanton effect. In addition, CP violation beyond SM is required for successful baryogenesis in the early universe.

It is thus of fundamental importance to measure the effects of CP violation in different ways. On particle colliders, we measure cross sections or differential decay rates which are functions of 3-momenta of outgoing particles. In such situations, CP violation usually manifests itself as P-odd combination of external momenta, which requires a totally anti-symmetric Levi-Civita tensor $\epsilon^{ijk}$. Therefore, we need at least three linearly independent momenta $k_{\alpha}~(\alpha=1,2,3)$ to contract all indices in $\epsilon^{ijk}$. As a consequence, $\epsilon^{ijk}k_{1i}k_{2j}k_{3k}$ is nonzero only when the three momenta are not coplanar. Given one further constraint of momentum conservation, this means that we need at least four external particles to form such a P-odd combination.

For instance, to study the CP properties of a heavy scalar $X$, one often uses the decay channel $X\rightarrow VV\rightarrow 4f$ where $V$ represents a gauge boson that subsequently decays into a pair of fermions $ff$. The momenta of the four final fermions can then form a P-odd function. This decay pattern has been evoked to study CP violation in the neutral pion decay \cite{Yang:1950rg}, 
$B$-physics \cite{Dunietz:1990cj,Kramer:1991xw}, and Higgs physics \cite{DellAquila:1985mtb,Soni:1993jc,Barger:1993wt,Choi:2002jk,Bolognesi:2012mm,Kovalchuk:2009zz,Hagiwara:2009wt,Artoisenet:2013puc}.

In this paper, we propose a new way to probe P- and CP-violating effects beyond the TeV-scale collider experiments, making use of the cosmic correlation functions of primordial fluctuations generated during inflation.

The idea of quasi-single field inflation and a cosmological collider has been proposed and studied in the recent years \cite{Chen:2009we,Chen:2009zp,Baumann:2011nk,Assassi:2012zq,Chen:2012ge,Noumi:2012vr,Pi:2012gf,Gong:2013sma,Arkani-Hamed:2015bza,Chen:2016nrs,Chen:2016uwp,Chen:2016hrz,Lee:2016vti,An:2017hlx,An:2017rwo,Kumar:2017ecc,Kumar:2018jxz,Wang:2018tbf,Lu:2019tjj,Alexander:2019vtb,Hook:2019zxa,Hook:2019vcn}. Assuming an inflationary background, the cosmological collider aims to study properties like the mass, spin and couplings of quantum fields in the very early universe by measuring correlation functions of primordial fluctuations. The energy scale of the cosmological collider, set by the Hubble constant during inflation $H\lesssim 10^{13-14}$ GeV, is much higher than earth-based colliders of any type in any foreseeable future. Its underneath idea is that particles produced in the exponentially expanding spacetime interact with each other and leave characteristic imprints on the correlation functions of curvature fluctuations and tensor fluctuations (primordial gravitational waves). For example, the mass of the mediator is encoded in the scaling/oscillation behavior of correlation functions in the soft limit, its spin is extracted from the angular dependence $P_S(\cos\theta)$ \cite{Alexander:2019vtb}, and the couplings are read from the size of non-Gaussianities $f_{NL}$ in the correlation functions. For light fields, signals are usually large \cite{Chen:2009zp}. For heavy fields, signals are in general suppressed by Boltzmann factors, but can be naturally lifted up in the presence of chemical potentials \cite{Chen:2018xck,Chua:2018dqh}, or in the scenario of special inflation models \cite{Flauger:2016idt,Tong:2018tqf}. Even for off-shell production of massive particles, the EFT description still gives a signal which is only power-law-suppressed \cite{Arkani-Hamed:2018kmz}. Therefore, with future experiments \cite{Dore:2014cca,Munoz:2015eqa,Meerburg:2016zdz,Ahmed:2014ixy,Suzuki:2015zzg,Abazajian:2016yjj,Matsumura:2013aja,Li:2017drr} on primordial non-Gaussianities and gravitational waves on the way, the possibility of utilizing this cosmological collider to probe particle physics at extremely high energy scales is tantalizing and promising.

In this paper we construct models to generate observably large CP-violating trispectrum, similar to the decay plane correlations in the $X\rightarrow VV\rightarrow 4f$ process. Our proof-of-concept calculation should be easily embedded into more realistic models. Our model borrows the structure of the Higgs-gauge sector of SM, and also makes use of a rolling axion-like field $\chi(t)$ that couples to a massive $U(1)$ gauge boson which we simply call $Z$ boson. The axion field can be either QCD axion or string axion or any axion-like particle. It is also possible to identify the axion as the inflaton. 

In our model, CP is spontaneously broken by the rolling background of the axion. Through the coupling between $Z$ to the inflaton fluctuation $\varphi$ and the Higgs field $h$ that gives mass to $Z$ and is derivatively mixed with the inflaton, the 4-point correlation function of $\varphi$ develops an imaginary part, signaling P-violation. Furthermore, the imaginary part is an odd function of the dihedral angle between two planes defined by four external momenta. In a parameter regime where loop expansions are trustworthy, the signal can reach up to $\tau_{NL}\sim\mathcal{O}(10^2)$. Stronger couplings and IR growth may further enhance the signal strength. Since the current observation by Planck 2013 \cite{Ade:2013ydc} gives $\tau_{NL}<2800$ (95\% CL), the signals in our model can be searched for in future surveys of primordial non-Gaussianities
.

In addition, by studying the large mass EFT description of our model, we conclude that in de Sitter (dS) spacetime, the infinite tower of local P- and CP-violating operators made of four inflatons are unobservable, since they contribute to a pure phase in the wavefunction of the universe. Thus the CP-violation signals in our model come from the on-shell particle production due to chemical potential and spacetime expansion. The extrapolation of the null result in exact dS to real inflationary scenario is a slow-roll suppressed signal. Hence perturbative unitarity and the smallness of slow-roll parameters should put a bound on the strength of the CP-violation signal in single field inflation EFT. A violation of this bound would indicate the non-local particle production effect present in quasi-single field inflation, along with potentially large cosmological collider signals.

Existing studies of P-violating bispectra and trispectra \cite{Kamionkowski:2010rb,Maldacena:2011nz,Soda:2011am,Shiraishi:2011st,Shiraishi:2012sn,Shiraishi:2013kxa,Shiraishi:2016mok,Barnaby:2012xt,Cook:2013xea,Bartolo:2017szm,Bartolo:2018elp,Bartolo:2019eac,Maleknejad:2016qjz,Maleknejad:2018nxz} require the presence of either tensor modes or broken rotational symmetry. These studies are thus aimed to probe the P or CP of the inflaton background or the gravitational interactions. Our construction is different in that none of these two ingredients are needed, and our motivation is to test the violation of P and CP in particle physics. 

This paper is organized as follows. We first briefly review in Sect.~\ref{SectXVV4f} the 4-body decay $X\rightarrow VV\rightarrow 4f$ of a heavy scalar $X$, and review in Sect.~\ref{SectCPinNG} the basics of primordial non-Gaussianity. We then show a CP-violating signal in primordial trispectrum resulting from a toy example in Sect.~\ref{SectToy}. We provide a more realistic model in Sect.~\ref{SectaZhphi} and study its properties for different chemical potential and mass choices. 
We conclude in Sect.~\ref{Conclusions}.

\section{Decay plane correlation in $X\rightarrow VV\rightarrow 4f$}\label{SectXVV4f}

We are mainly interested in the CP-violating effects in the correlation functions of primordial density fluctuations, $i.e.$, the cosmological collider observables. Before a more systematic study of this topic, it is useful to recall how to probe CP-violating effects in the decay of a heavy scalar particle $X$ on a collider through the channel $X\to VV\to 4f$. Following \cite{Bolognesi:2012mm,Kovalchuk:2009zz}, we give a concise review of this useful decay channel. In FIG.\;\ref{momentaConfig} we show the Feynman diagram of this process in the left panel and the momentum configuration of the four final fermions in the right panel. Consider first the decay of $X$ with 4-momentum $p$ into two off-shell gauge bosons $VV$ with 4-momenta $q_1,q_2$ and polarization vectors $\epsilon_1,\epsilon_2$. The most general amplitude for this process respecting Lorentz symmetry is
\begin{equation}\label{XVVparametrization}
	\mathcal{A}(X\rightarrow VV)=\mathcal{F}_1 \epsilon_1^*\cdot\epsilon_2^*+\frac{\mathcal{F}_2 }{m_X^2}(\epsilon_1^*\cdot p)(\epsilon_2^*\cdot p)+i\frac{\mathcal{F}_3}{m_X^2}\epsilon^{\mu\nu\rho\sigma}p_\mu P_\nu\epsilon^*_{1\rho} \epsilon^*_{2\sigma}~,
\end{equation}
where 
$P=p_1-p_2$. 
The first two terms stand for P-even $S$-wave contribution and $D$-wave contribution, while the last term is the P-odd $P$-wave amplitude. The form factors $\mathcal{F}_i~(i=1,2,3)$ are functions of momentum squares. It is easy to identify the EFT operators corresponding to these form factors. At the leading order of a gradient expansion of the EFT operators, we have $\mathcal{F}_1  \leftrightarrow X V_\mu V^\mu$, $\mathcal{F}_2\leftrightarrow X V_{\mu\nu} V^{\mu\nu}$ and $\mathcal{F}_3\leftrightarrow X V_{\mu\nu} \tilde{V}^{\mu\nu}$.

The vector bosons subsequently decay to Dalitz pairs, with an amplitude given by
\begin{equation}
\label{AXVV4f}
	\mathcal{A}(X\rightarrow VV\rightarrow 4f)=\mathcal{A}_{\rho\sigma} (X\rightarrow VV) D_V^{\rho\alpha}(\bar{u}\Gamma_\alpha v)D_V^{\sigma\beta}(\bar{u}\Gamma_\beta v)~,
\end{equation}
with $\Gamma_\mu$ being the 1PI vertex of the fermion-vector interaction. Focusing on the $P$-wave contribution, we see
\begin{equation}
	\mathcal{A}(X\rightarrow VV\rightarrow 4f)|_\text{$P$-wave}=i\frac{\mathcal{F}_3}{m_X^2}\epsilon_{\mu\nu\rho\sigma}p^\mu q^\nu D_V^{\rho\alpha}(\bar{u}\Gamma_\alpha v)D_V^{\sigma\beta}(\bar{u}\Gamma_\beta v)~.
\end{equation}
To further simplify this expression we focus on the transverse polarizations of the vector boson and write $\bar{u}(q_1)\Gamma_\alpha v(q_2)\sim (q_1-q_2)_\alpha$. Consequently,
\begin{equation}
	\mathcal{A}(X\rightarrow VV\rightarrow 4f)|_\text{$P$-wave}\sim i\frac{\mathcal{F}_3}{m_X^2}\epsilon_{\mu\nu\rho\sigma}(p_1+p_2)^\mu (p_1-p_2)^\nu (q_1-q_2)^\rho (k_1-k_2)^\sigma~.
\end{equation}
Here we used the antisymmetry of the Levi-Civita symbol to simplify the tensor structure of the vector propagator. After boosting to the rest frame of the scalar, $p=p_1+p_2=(m_X,0,0,0)$ only has a time-like component and the amplitude is proportional to the rotational invariant
\begin{equation}\label{flatAmp}
\mathcal{A}(X\rightarrow VV\rightarrow 4f)|_\text{$P$-wave}\propto\epsilon^{ijk}(q_1+q_2)_i(q_1-q_2)_j(k_1-k_2)_k~.
\end{equation}
As a result, with spin degrees of freedom neglected, decaying through $P$-wave channel is forbidden if $\vec{q}_1$, $\vec{q}_2$ and $\vec{k}_1$ are coplanar. Around this planar configuration, the amplitude behaves as an odd function of the dihedral angle between the two decay planes spanned by $\vec{q}_1,\vec{q}_2$ and $\vec{k}_1,\vec{k}_2$. To show this more explicitly, we parametrize $\vec{q}_1=(q_1 \sin\theta_q,0,-q_1 \cos\theta_q)$, $\vec{q}_2=(-q_1 \sin\theta_q,0,-p+q_1 \cos\theta_q)$, $\vec{k}_1=( k_1\sin\theta_k\cos\phi,k_1 \sin\theta_k\sin\phi,k_1\cos\theta_k)$, $\vec{k}_2=(-k_1 \sin\theta_k\cos\phi,-k_1 \sin\theta_k\sin\phi,p-k_1\cos\theta_k)$ in FIG.~\ref{momentaConfig}. The amplitude is then proportional to 
\begin{equation}
	(\vec{q}_1+\vec{q}_2)\cdot[(\vec{q}_1-\vec{q}_2)\times(\vec{k}_1-\vec{k}_2)]=-4q_1k_1p \sin\theta_k\sin\theta_q\sin\phi~.
\end{equation}
\begin{figure}[h!]
	\centering
	\includegraphics[width=6cm]{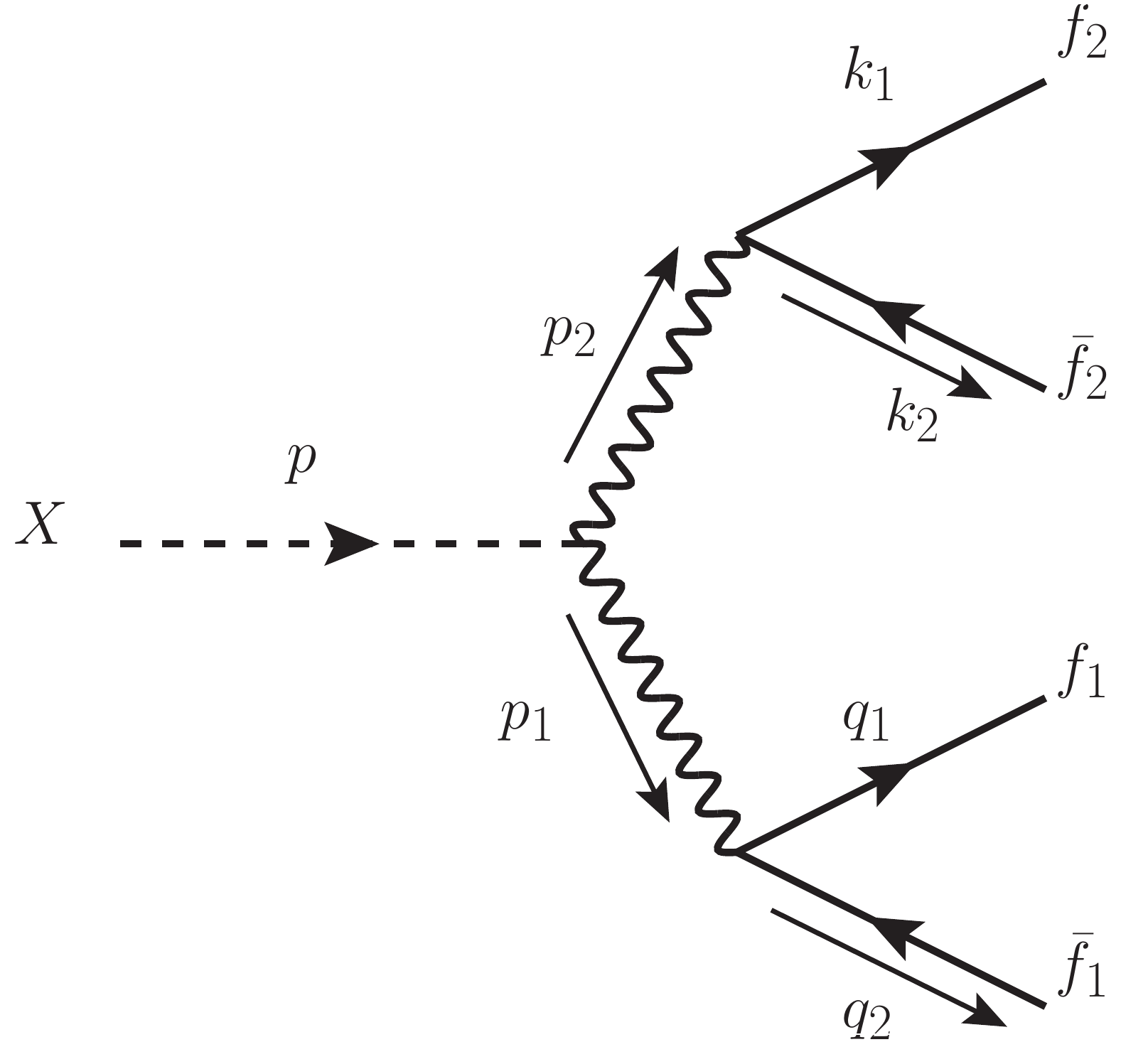}~~~~~~~~~~~~~~~~~~~~~~~~~~~~~\includegraphics[width=6cm]{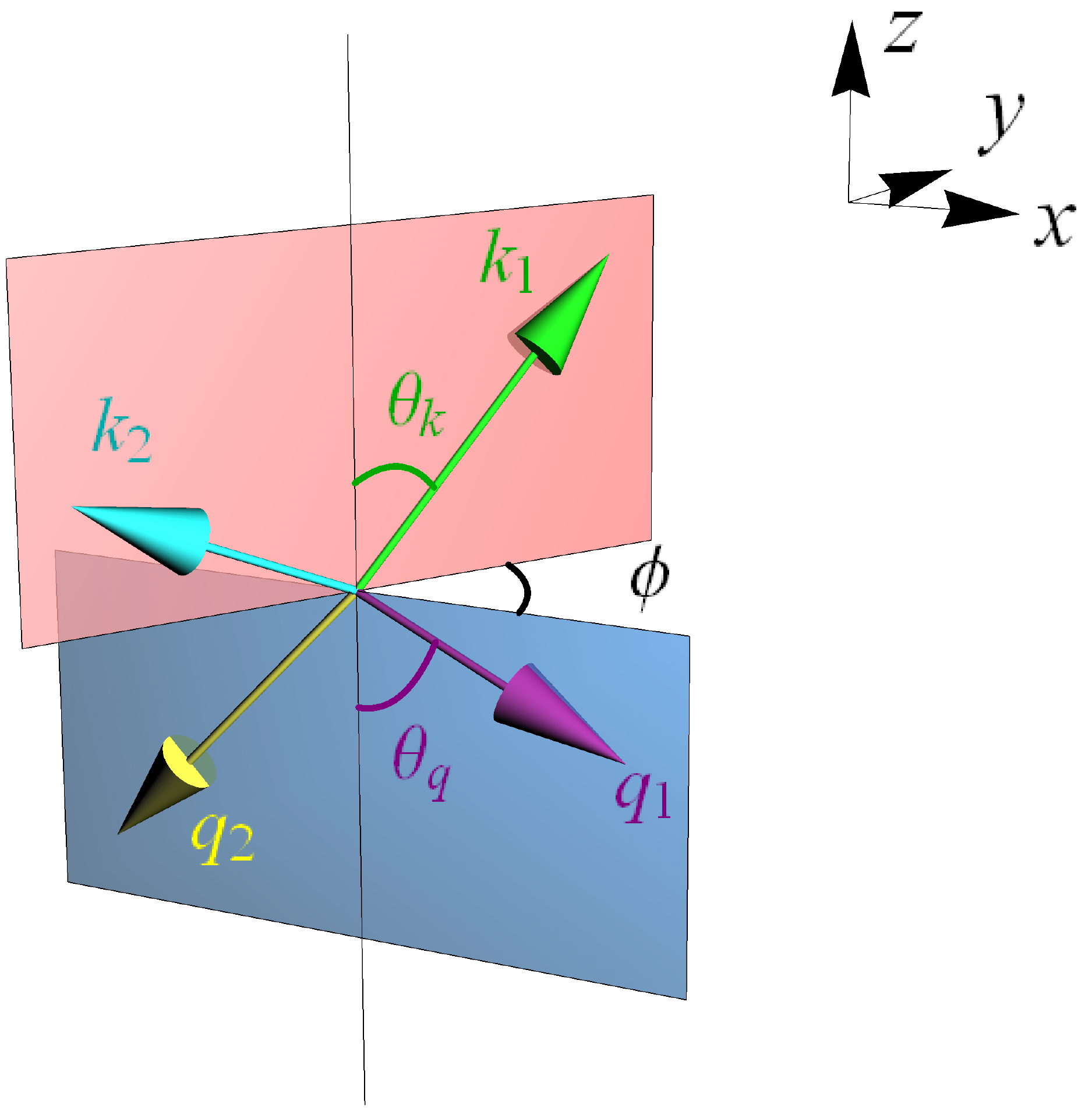}
	\caption{The kinematics of $X\rightarrow VV\rightarrow 4f$.}\label{momentaConfig}
\end{figure}
If the dynamical factor is regular around $\sin\phi=0$, the amplitude is odd in $\sin\phi\leftrightarrow -\sin\phi$, with $\phi$ being the relevant dihedral angle. As we mentioned before, the P-odd shape vanishes when all momenta are coplanar $(\phi=0)$. 

To monitor the $\phi$ dependence in the correlation of decay products, we calculate the differential decay rate $\di\Gamma/\di\phi$, which is the modular square of the amplitude (\ref{AXVV4f}), with all final-state spin states summed, and all phase space variables except $\phi$ integrated,
\begin{eqnarray}
	\nonumber\frac{d\Gamma_{X\rightarrow VV\rightarrow 4f}}{d\phi}&=&\int d\Pi_{k_1,q_1,p,\theta_k,\theta_q}\sum \Big|\mathcal{A}|_\text{$S$-wave}+\mathcal{A}|_\text{$D$-wave}+\mathcal{A}|_\text{$P$-wave}\Big|^2\\
	&=&\int d\Pi_{k_1,q_1,p,\theta_k,\theta_q}\sum |A+B\sin\phi|^2=(\text{P-even})+(\text{P-even})\sin\phi~.
\end{eqnarray}
In the second line, we have explicitly spelled out the $\sin\phi$ dependence in the P-wave amplitude, and the coefficients $A$ and $B$ are both symmetric under $\phi\to-\phi$. In the final result, therefore, we get a piece odd in $\phi\to-\phi$ which arises from the interference between the P-even ($S$ and $D$-wave) amplitude and P-odd ($P$-wave) amplitude. The P-odd dependence in the final result is nonzero only when $AB\neq 0$, namely when both CP-even and CP-odd pieces are present. Therefore the P-odd behavior is a signature of CP violation. We note in particular that the case with $A=0$ and $B\neq 0$ will not generate a P-odd shape in the final result. 



\section{A Brief review of Primordial Non-Gaussianity}
\label{SectCPinNG}

In this section we provide a very brief review of primordial non-Gaussianity in the context of cosmological collider physics. We refer readers to \cite{Chen:2017ryl} for a pedagogical review.

In ordinary inflation scenarios, the primordial fluctuations as we observe today were generated from the quantum fluctuation of the inflaton field $\phi=\phi_0+\varphi$ during inflation\footnote{Notice that the inflaton field $\phi$ is to be distinguished from the dihedral angle mentioned above.}. We use $\phi_0$ to denote the background and $\varphi$ the fluctuation. Due to the flatness of the inflation potential, the fluctuation $\varphi$ behaves approximately as a massless scalar field. It is convenient to expand this fluctuation field in terms of Fourier modes $\varphi(\tau,\vec{k})$, and the mode function of positive frequency part is given by
\begin{equation}
\label{dephi}
  \varphi(\tau,\vec{k})=\frac{H}{\sqrt{2k^3}}(1+\ii k\tau)e^{-\ii k\tau},
\end{equation}
where $H$ is the Hubble parameter during inflation and $\tau$ is the conformal time which goes from $-\infty$ to 0. At the late time limit $\tau\to 0$, the mode function approaches to a constant $\varphi(0,\vec{k})=H/\sqrt{2k^3}$, and thus provides the initial condition for the density fluctuations of our universe. 

From the observation we can probe the $n$-point correlation functions of $\varphi$, $i.e.$, $\la\varphi(\tau,\vec{k}_1),\cdots \varphi(\tau,\vec{k}_n)\ra$. The 2-point function gives the power spectrum which is well-measured today. The late time limit of the mode function (\ref{dephi}) predicts a scale invariant power spectrum $\la\varphi^2\ra\sim H^2/(2k^3)$, which is correct up to slow-roll corrections. Higher point correlations ($n\geqslant 3$) provide the information about the interactions of $\varphi$ modes, and are known as primordial non-Gaussianities. Therefore, we can think of non-Gaussianity effectively as ``inflaton collisions''. Upon a gauge transformation $\zeta=-(H/\dot\phi_0)\varphi$ that translates inflaton fluctuations into curvature fluctuations $\zeta$, the inflaton correlations can also be expressed in terms of $\zeta$, which is also widely used. Here $\dot\phi_0$ is the rolling speed of the inflaton background and can be approximated as a constant during inflation for our calculation. 

The inflaton field is not the only field present during inflation. The fast expansion of the inflationary universe allows for the production of any heavy particles with mass up to $\order{H}$. Heavy particles decay quickly in an expanding background and cannot survive the late time limit. However, they may interact with the inflaton fluctuation $\varphi$ before they decay, leaving imprints on the non-Gaussianity. Therefore, measuring non-Gaussianity can be viewed as a way to extract information of short-lived heavy particles by monitoring the correlation of long-lived inflaton modes. This is completely in parallel with the logic of modern collider experiments, and for this reason, this approach is called the cosmological collider physics. 

However, we should also point out a major distinction from collider experiments. In collider experiments, it is usually possible to reconstruct the phase space of a process at the single-event level. On the cosmological collider, single-event level signals are typically overwhelmed by large fluctuations in the IR and are therefore undetectable. Instead, we can only measure the statistical average of a large number of events from correlation functions. In this regard, the cosmological collider is less informative than ordinary collider experiments.

The techniques of computing non-Gaussianity is quite similar to usual calculation of $S$-matrix, with some complication from the curved spacetime background. In particular, we can still use a diagrammatic approach to organize the perturbative expansion. The key difference is the lack of explicit in and out states in the case of cosmic correlators. As a result, we should calculate Schwinger-Keldysh diagrams rather than ordinary Feynman diagrams, although the ``Feynman rules'' are similar. The rules not only allow us to calculate the non-Gaussianty explicitly, but also provide a way to estimate the size of the result before a detailed calculation. 

In this paper we are mostly interested in the 4-point correlation of the inflaton fluctuations, since the CP violation in general implies a shape containing Levi-Civita symbol $\ep^{ijk}$. As explained above, we need at least three independent momenta to form a nonzero result when contracted with $\ep^{ijk}$. This means that we need to consider at least 4-point correlations, since the external momenta of $n$-point correlations are subject to momentum conservation. 

At the 4-point level, it is customary to parameterize the correlation function by the trispectrum as
\begin{align}
\label{trispec}
  \la\zeta(\vec{k}_1)\zeta(\vec{k}_2)\zeta(\vec{k}_3)\zeta(\vec{k}_4)\ra'=(2\pi)^6P_\zeta^3\frac{K^3}{(k_1k_2k_3k_4)^3}T(\vec{k}_1,\vec{k}_2,\vec{k}_3,\vec{k}_4),
\end{align}
where $K=\sum_{i=1}^4 k_i$, and the prime on $\la\cdots\ra'$ means the $\delta$-function of momentum conservation is removed. Here $P_\zeta$ is the power spectrum defined via $(2\pi^2/k^3)P_\zeta=\la\zeta(\vec k)\zeta(-\vec k)\ra'$ and is measured to be $P_\zeta\simeq 2\times 10^{-9}$ at the CMB scale. Again using the conversion $\zeta=-(H/\dot\phi_0)\varphi$, we can find an estimate of the trispectrum $T$ as,
\begin{equation}
  T\sim\frac{P_\zeta^{-1}}{(2\pi)^2}\times\left(\text{loop factors}\right)\times\left(\text{vertices}\right)\times\left(\text{propagators}\right).
\end{equation}
Here again every dimensional parameter is measured in the unit of the Hubble parameter $H$, as long as the particles are not far heavier than the scale $H$. The four external momenta subject to momentum conservation can form non-coplanar configuration and thus we can look for signals of CP violation from the trispectrum. In the next section we use a toy example to illustrate how this can be achieved.

\section{A toy model: Scalar QED in de Sitter}\label{SectToy}

To draw analogy to the collider event in FIG.~\ref{momentaConfig}, we now consider a diagram contributing to the 4-point cosmic correlator with exactly the same topology, shown in FIG.~\ref{ToyDiagram}. In this figure, we are imagining a pair of scalar fields $\phi^\pm$ (thick dashed lines with arrows) being external lines. These fields are charged under a $U(1)$ gauge group, and thus interact with the $U(1)$ gauge field $A_\mu$ (wiggly lines) through minimal coupling as in a scalar QED. Unlike the collider process in FIG.~\ref{momentaConfig}, here we do not have an initial heavy scalar particle $X$. Instead, we can introduce a CP-odd operator insertion $\theta(\tau)F\wt F$. Here $\theta(\tau)$ has explicit time dependence and, as we shall see below, behaves effectively as a particle source producing gauge bosons $A_\mu$ during inflation. Such a $\theta$ term can be easily generated from a coupling to an axion-like field $\chi$ through the term $\chi F\wt F$, by allowing the background of $\chi$ to slowly roll down its potential during inflation.

\begin{figure}[h!]
	\includegraphics[width=6cm]{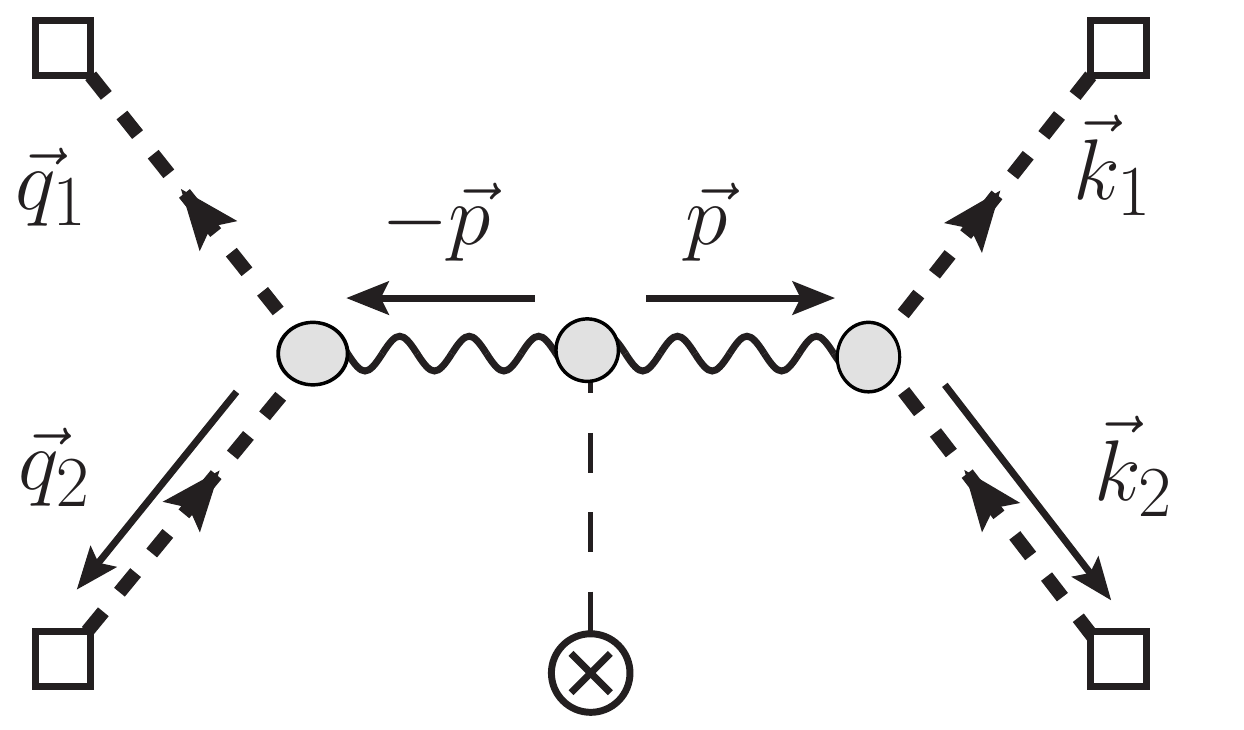}\\
	\caption{The induced CP-violating t-channel diagram. See the text for explanation of various lines. We adopted the diagrammatic notation in \cite{Chen:2017ryl}.} \label{ToyDiagram}
\end{figure}

The model generating the above diagram is simply the scalar QED with an additional time-dependent $\theta$-term. Its action is
\begin{equation}
	S=\int d\tau d^3x \sqrt{-g}\left[-g^{\mu\nu}D_\mu\phi^*D_\nu\phi-m^2\phi^*\phi-\frac{1}{4}g^{\mu\rho}g^{\nu\sigma}F_{\mu\nu}F_{\rho\sigma}-\frac{1}{4}c_0\theta(\tau) \mathcal{E}^{\mu\nu\rho\sigma}F_{\mu\nu}F_{\rho\sigma}\right]~,
\end{equation}
where $\mathcal{E}^{\mu\nu\rho\sigma}=\epsilon^{\mu\nu\rho\sigma}/\sqrt{-g}$ is the covariant Levi-Civita tensor. The $\theta$-term is dynamical since it is explicitly time-dependent. Evaluating the above action on the inflationary background with the spacetime metric $g_{\mu\nu}=a^2(\tau)\eta_{\mu\nu}$ with $a\simeq-1/(H\tau)$, we have,
\begin{eqnarray}
	\nonumber S &=&\int d\tau d^3x\Big[-a^2\eta^{\mu\nu}\partial_\mu\phi^*\partial_\nu\phi-m^2a^4\phi^*\phi-\frac{1}{4}\eta^{\mu\rho}\eta^{\nu\sigma}F_{\mu\nu}F_{\rho\sigma}\\
	&&~~~~~~~~~~~~~+ie a^2\eta^{\mu\nu} \phi^*\overleftrightarrow{\partial_\mu}\phi A_\nu-e^2 a^2\phi^*\phi A_\mu A_\nu \eta^{\mu\nu}-\frac{1}{4}c_0\theta(\tau) \epsilon^{\mu\nu\rho\sigma}F_{\mu\nu}F_{\rho\sigma}\Big]~.
\end{eqnarray}

In this toy example we neglect the back-reaction of the quantum fields on the spacetime geometry. In addition, we require $\langle\phi\rangle=0$ to keep   gauge invariance manifest. Upon integration by part, the last term takes the form of a Chern-Simons term with a time-dependent factor in the front. 
\begin{equation}
	-\int d\tau d^3x\frac{\theta(\tau)}{4} \epsilon^{\mu\nu\rho\sigma}F_{\mu\nu}F_{\rho\sigma}=\int d\tau\theta'(\tau)\int d^3x \epsilon^{ijk}A_i\partial_j A_k~.
\end{equation}
This is essentially the spatial integral of the Chern-Simons charge density $J_{CS}^0=\epsilon^{ijk}A_i F_{jk}$ \cite{Kharzeev:2009fn}. In other words, for $c_0 \theta'>0$ ($c_0 \theta'<0$) the rolling background pumps left-handed (right-handed) states out of the vacuum while destroying right-handed (left-handed) states into the vacuum. This can be viewed as the physical source of P-violation. 
We are interested in the decay of the photon pair into scalars, therefore no special care about boundary condition is needed. 

To pursue a perturbative calculation, 
we quantize the system using the Schwinger-Keldysh path-integral formalism \cite{Chen:2017ryl}. The relevant Feynman rules are given in Appendix ~\ref{FeynmanRulesAppendix}.

The t-channel diagram (shown in FIG.~\ref{ToyDiagram}) is given by

\begin{eqnarray}
	\nonumber\left\langle\phi_{\vec{q}_1}\phi^*_{\vec{q}_2}\phi_{\vec{k}_1}\phi^*_{\vec{k}_2}\right\rangle'_{t}&=&-e^2 c_0\sum_{\epsilon_2=\pm}\epsilon_2\int d\tau_1 d\tau_2 d\tau_3 a(\tau_1)^2 \theta'(\tau_2) a(\tau_3)^2\epsilon^{ijk}(q_1-q_2)_i p_j(k_1-k_2)_k\\
	\nonumber&&~~~~~~~~~~~\times G_{\epsilon_1}(q_1,\tau_1)G_{\epsilon_1}^*(q_2,\tau_1)D_{\epsilon_1\epsilon_2}(p,\tau_1,\tau_2)D_{\epsilon_2\epsilon_3}(p,\tau_2,\tau_3)G_{\epsilon_3}(k_1,\tau_3)G_{\epsilon_3}^*(k_2,\tau_3)\\
	&\equiv&(\vec{q}_1+\vec{q}_2)\cdot[(\vec{q}_1-\vec{q}_2)\times(\vec{k}_1-\vec{k}_2)]~F(\vec{q}_1,\vec{q}_2,\vec{k}_1,\vec{k}_2)~,
\end{eqnarray}
where $G$ is the propagator for a massive scalar in dS and $D$ denotes the propagator of a massless vector in flat spacetime with the tensor structure stripped (because gauge fields are conformally coupled in four-dimensional spacetime). The u-channel contribution is obtained by exchanging two anti-particles,
\begin{equation}
	\left\langle\phi_{\vec{q}_1}\phi^*_{\vec{q}_2}\phi_{\vec{k}_1}\phi^*_{\vec{k}_2}\right\rangle'_{u}=\left\langle\phi_{\vec{q}_1}\phi^*_{\vec{k}_2}\phi_{\vec{k}_1}\phi^*_{\vec{q}_2}\right\rangle'_{t}~.
\end{equation}
There is no $s$-channel contribution in this toy model since we can distinguish two charge eigenstates for external scalars. 
Therefore the total CP-odd contribution to the 4-point function is
\begin{eqnarray}\label{ToyAmp}
	\nonumber\left\langle\phi_{\vec{q}_1}\phi^*_{\vec{q}_2}\phi_{\vec{k}_1}\phi^*_{\vec{k}_2}\right\rangle'&=&2\left[\left\langle\phi_{\vec{q}_1}\phi^*_{\vec{q}_2}\phi_{\vec{k}_1}\phi^*_{\vec{k}_2}\right\rangle'_{t}+\left\langle\phi_{\vec{q}_1}\phi^*_{\vec{q}_2}\phi_{\vec{k}_1}\phi^*_{\vec{k}_2}\right\rangle'_{u}\right]\\
	&=&2(\vec{q}_1+\vec{q}_2)\cdot[(\vec{q}_1-\vec{q}_2)\times(\vec{k}_1-\vec{k}_2)]\left(F(\vec{q}_1,\vec{q}_2,\vec{k}_1,\vec{k}_2)-F(\vec{q}_1,\vec{k}_2,\vec{k}_1,\vec{q}_2)\right)~,
\end{eqnarray}
which bears an analogues form as its flat spacetime ancestor (\ref{flatAmp}). When the 4-point function is observed for a momentum set up as in FIG.~\ref{momentaConfig}, the trispectrum will acquire a $\sin\phi$ dependence that apparently violates parity conservation.

Some conceptual remarks are in order before we conclude this section.

First, in flat spacetime, the combined transformation CPT is an unbroken discrete symmetry as a consequence of Poincar\'e symmetry. Thus any CP violation in the theory is equivalently a T violation, and vice versa. However, the CPT theorem does not hold in its original form in curved spacetime. It is easy to find models violating T but preserving CP. For example, consider a minimally coupled real scalar in an expanding spacetime. The time reversal symmetry T is spontaneously broken by the expansion of the universe, so is the combined CPT. Hence CP violation is not automatic in an expanding universe, instead, it is model-dependent. 

Second, we claimed that the CP-odd shape in the trispectrum is a signal of CP violation and we have been comparing this signal with the CP-odd signal in the differential decay rate of $X\to 4f$ on a particle collider. But there is a subtlety in making this comparison. In cosmology we observe correlation functions instead of cross sections or decay rates. The cross sections are always modular squares of scattering amplitudes, which means that the appearance of Levi-Civita symbol in cross sections (or differential decay rates) must be from the interference between a CP-even piece and a CP-odd piece at the amplitude level. The existence of both CP-even and CP-odd pieces with the same outgoing states is usually required to establish the CP-violating signal. However, in cosmic correlators, there exists interference between process with different outgoing states. In particular, there is \textit{always} a CP-even Gaussian piece standing for free propagation, with which the CP-odd non-Gaussian piece can interfere. Thus in a sense, we observe amplitudes directly instead of their modular squares. Consequently, we can see Levi-Civita symbols directly at the amplitude level. This does not necessarily imply CP violation at the Lagrangian level, since we can assign $CP=-1$ to the scalar field (the axion-like field $\chi$ in our toy example). Nonetheless, we need a nonzero and rolling background of $\chi$ to generate the desired $\theta(\tau)$ term in this toy model, and therefore the CP is \emph{spontaneously} broken by the rolling classical profile of the axion field. So we can still say that the CP-odd shape in the trispectrum in our model is a signal of CP violation.
 
Third, a technical point to be made is that we require a pair of charged scalar fields appearing in external lines, not only to couple them to a $U(1)$ gauge boson, but also to make the four external lines non-identical. This is important to generate a nonzero CP-odd shape, since if we choose the four external lines to be identical, the combination 
in (\ref{ToyAmp}) will vanish.

\section{CP-violating trispectrum in a realistic model}\label{SectaZhphi}

In the last section we realized a CP-violating trispectrum in a toy model made of spectator fields on dS background. This result cannot be applied directly to generic inflation models because, as we mentioned above, we need two scalars with opposite $U(1)$ charge in external lines, to generate a nonzero CP-odd trispectrum. A pair of two distinguishable charged scalars are not available in minimal inflation scenarios, since that results in large isocurvature fluctuations which have been severely constrained by observations. On the other hand, if we simply replace all external lines by inflaton fluctuations, the CP-violating shapes will be canceled due to permutations (cf.\ the end of last section). Therefore, we need to find other ways to construct two distinguishable external lines. 

To this end, we still need at least two real scalar fields $\phi$ and $\si$. We will assume that $\phi$ is the inflaton field which has very flat potential and its fluctuation $\varphi$ is nearly massless. On the other hand, $\si$ is in general massive. To realize a process with topology similar to the previous example, we require a coupling $\partial_\mu\phi Z^\mu\sigma$ and also a two-point mixing $\dot{\phi}\sigma$ that converts the second scalar $\sigma$ to the observable $\phi$. These two couplings can appear naturally in an effective theory of inflaton plus a $U(1)$ gauge sector with a complex Higgs scalar field\footnote{Note that our model takes the form of the electroweak sector of SM, yet they are not necessarily the same. To be general, we choose to consider the fields in our model as BSM fields hereafter and comment later on the special case where they are the SM fields.}. Among the leading inflaton-matter couplings we have the following operators \cite{Kumar:2017ecc},
\begin{equation}
\mathcal{L}_{\text{int}}^{\text{inf-gauge}}=\frac{c_1}{\Lambda}\partial_\mu\phi(\mathcal{H}^\dagger D^\mu\mathcal{H})+\frac{c_2}{\Lambda^2}(\partial\phi)^2\mathcal{H}^\dagger\mathcal{H}+\cdots~.
\end{equation}
Upon symmetry breaking, either by input or by heavy-lifting, the $c_1$ term yields a triple vertex $\Delta\mathcal{L}_1=\frac{\rho_{1,Z}}{\dot{\phi}_0}h\partial_\mu \varphi Z^\mu$ that acts as a current-potential interaction, and the $c_2$ term results in a mixing between Higgs and inflaton $\Delta\mathcal{L}_2=-\rho_2\dot{\varphi}h$. The couplings are related to the EFT Wilson coefficients by $\rho_{1,Z}\equiv -\Im c_1\dot{\phi}_0 m_Z/\Lambda$ and $\rho_2\equiv 2c_2\dot{\phi}_0 v/\Lambda^2$.

In addition to the above two couplings, we also introduce the CP-violating interaction $\Delta\mathcal{L}_3=-\frac{c_0\theta(t)}{4}Z_{\mu\nu}Z_{\rho\sigma}\mathcal{E}^{\mu\nu\rho\sigma}$, where $\theta(t)$ is dependent on time. The CP violation induced by this operator will produce signatures in the trispectrum of $\varphi$. 

To summarize, with the inflaton background, our model consists of the following three couplings, 
\begin{equation}\label{AZHIEFTop}
	\Delta\mathcal{L}_1=\frac{\rho_{1,Z}}{\dot{\phi}_0}h\partial_\mu \varphi Z^\mu,~~~~~~\Delta\mathcal{L}_2=-\rho_2\dot{\varphi}h,~~~~~~\Delta\mathcal{L}_3=-\frac{c_0}{4}\theta(t)Z_{\mu\nu}Z_{\rho\sigma}\mathcal{E}^{\mu\nu\rho\sigma}~.
\end{equation}
As in the previous toy model, the time-dependent $\theta$-term is most naturally realized with a rolling axion field $\chi=f\theta$ where $f$ is the decay constant. In this case, CP is preserved at the Lagrangian level since the axion is CP odd, but CP is spontaneously broken by the rolling background of $\chi$. The dynamics of its classical profile is governed by the equation of motion
\begin{equation}\label{axionEOM}
	\ddot{\theta}+3H \dot{\theta}+\frac{V'(\theta)}{f^2}+\frac{c_0}{4f^2}\langle Z_{\mu\nu}Z_{\rho\sigma}\mathcal{E}^{\mu\nu\rho\sigma}\rangle=0~,
\end{equation}
where $V(\theta)=\Lambda_\chi^4\left(1-\cos\theta\right)$ is the axion potential. We require the energy density of the axion to be much smaller than that of the inflaton, $\Lambda_\chi^4\ll M_p^2 H^2$, to avoid multi-field inflation (for keeping things simple). The last term of (\ref{axionEOM}) comes from the back-reaction of perturbations on the background. Although in its apparent form this term appears to be a correction to the slope of the potential, it is actually proportional to the rolling speed $\dot{\theta}$ and thus serves as a frictional force $\Gamma\dot{\theta}\sim \frac{c_0}{4f^2}\langle Z_{\mu\nu}Z_{\rho\sigma}\mathcal{E}^{\mu\nu\rho\sigma}\rangle$. This is the usual dissipative effects due to particle production. The large friction produced by the combination of exponential spacetime expansion and dissipative effects tends to drive the axion to the slow-roll attractor phase, where $|\frac{\ddot{\theta}}{(3H+\Gamma)\dot{\theta}}|\ll 1$ and $\dot\theta\sim\frac{\Lambda_\chi^4}{f^2(3H+\Gamma)}\sim \text{const}$. For our purpose, it is convenient to absorb the axion rolling speed into the coupling constant and define an dimensionless parameter $c$ as
\begin{equation}
\label{cdef}
  c\equiv \frac{c_0\dot\theta}{H}\sim \frac{c_0 \Lambda_\chi^4}{f^2 (3H+\Gamma)H}~.
\end{equation}

We point out that the $\theta$ term is only P-violating by itself since it is C-invariant, hence breaking CP also. Even in the absence of a direct coupling between fermions and our axion, the time-dependent $\theta$ still provides a chemical potential for the fermion sector and thus brings extra P-violation that cannot be balanced by any C-violation. Consider the coupling of the $Z$ field to a fermion,
\begin{equation}
	\Delta\mathcal{L}_f=\bar{\psi}(i\slashed{D}-m)\psi-\frac{c_0}{4}\theta(t)Z_{\mu\nu}Z_{\rho\sigma}\mathcal{E}^{\mu\nu\rho\sigma}~.
\end{equation}
If $\theta(t)=\text{const}$, we can perform a global chiral redefinition $\psi'\equiv \exp\left[-i\alpha\gamma_5\right]\psi$ with $\alpha\propto\theta/2$, to eliminate the total derivative term, and also by doing so giving an invariant definition of fermion parity. However, if the $\theta$ term is dependent on time, there is no global chiral field redefinition that can eliminate the $\theta$ term once and for all. And the natural parity defined at one moment will differ from that of the next. Thus a chemical potential term for fermions will be induced at one-loop level, which is proportional to the rate of change of $\theta$:
\begin{equation}
	\begin{gathered}
	\includegraphics[width=3cm]{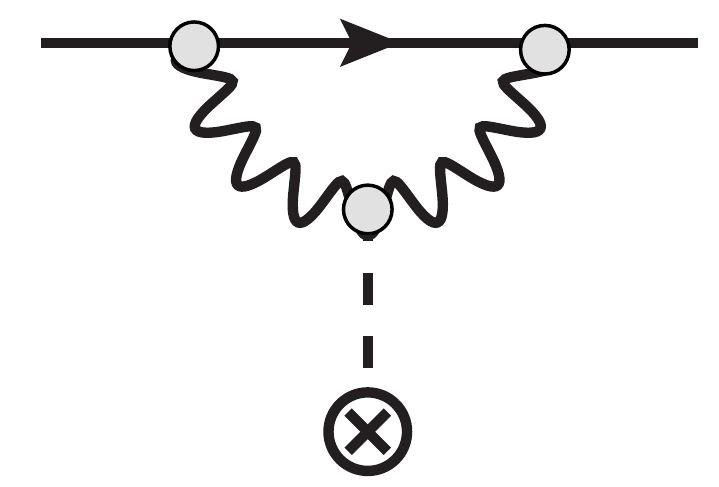}
	\end{gathered}\Rightarrow~\bar{\psi}\gamma_0\gamma_5\psi \partial_0\theta\propto c(n_R-n_L)~.
\end{equation}
This term is C-invariant but P-odd for fixed axion background. Using the EOM method described in Sect.~\ref{EOMmethod}, this simply corresponds to the one-loop self-energy for the fermion, with odd number of $\theta$ insertions contributing to the chemical potential term and even number of $\theta$ insertions contributing to the mass term and the field-strength renormalization factor. Note that in the SM setup, $U(1)_{B+L}$ is anomalous with respect to $SU(2)_L$ but not $U(1)_Y$. Thus only left-handed fermions are relevant to the induced chemical potential term $cn_{B+L}$.

\subsection{Leading-order perturbation theory}\label{1PTmethod}
For a perturbatively small $c\ll 1$, we can simply calculate the trispectrum to the leading order. This corresponds to the parameter regime where $\frac{c_0 \Lambda_\chi^4}{f^2 (3H+\Gamma)H}\ll 1$. Namely either the coupling $c_0$ is small or the axion rolling speed is slow. 

Again we can quantize the system using Schwinger-Keldysh formalism. The relevant Feynman rules are given in Appendix~\ref{FeynmanRulesAppendix}. Note that symmetry breaking gives $Z$ boson a mass and its different polarization modes have different EOMs. Since the $\theta$ factor already occupies the time-like component, only the spatial components of the gauge field propagator contribute. The longitudinal polarization is proportional to $\hat{p}_i \hat{p}_j$ and thus vanishes upon contraction with the Levi-Civita symbol. The transverse polarization is proportional to $\delta_{ij}-\hat{p}_i\hat{p}_j$ and is filtered to $\delta_{ij}$ by the same reasoning.

\begin{figure}[h!]
	\includegraphics[width=6cm]{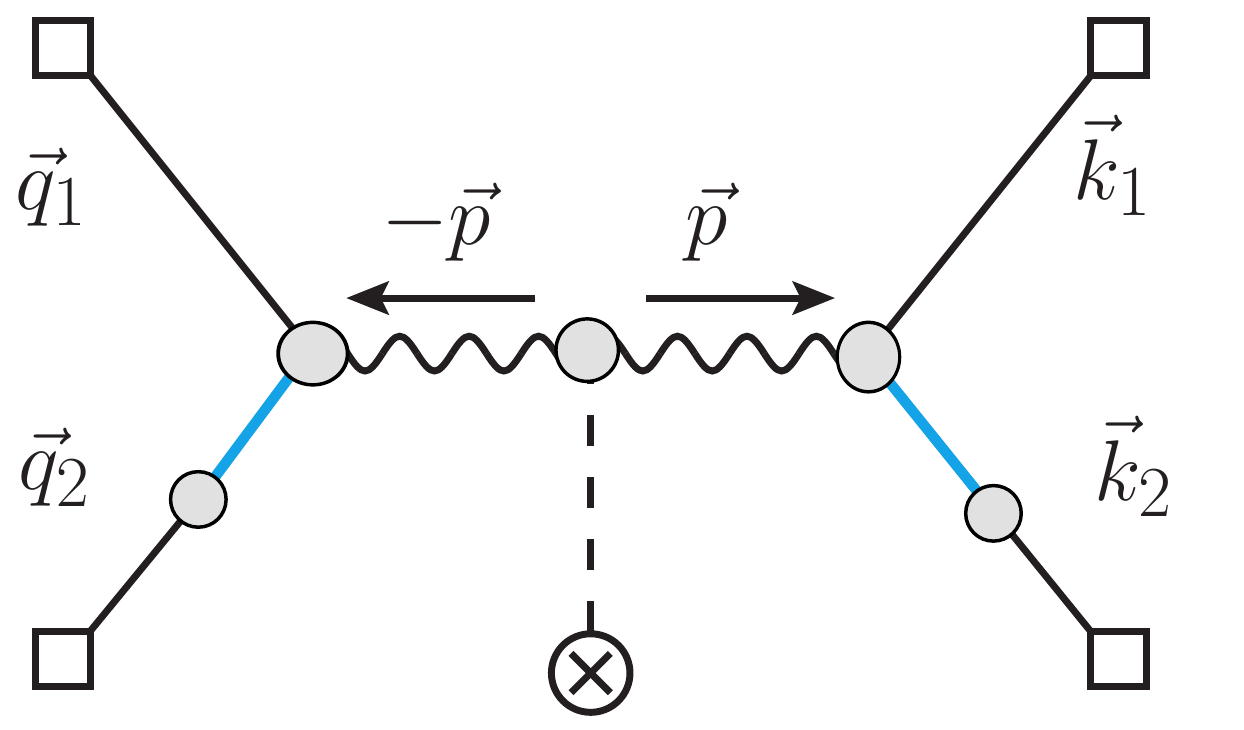}\\
	\caption{The leading-order CP-violating t-channel diagram.}\label{RealDiagram}
\end{figure}

The $t$-channel diagram is given by
\begin{eqnarray}
	\nonumber\left\langle\varphi_{\vec{q}_1}\varphi_{\vec{q}_2}\varphi_{\vec{k}_1}\varphi_{\vec{k}_2}\right\rangle'_{t}&=&-c_0\left(\frac{\rho_{1,Z}}{\dot{\phi}_0}\right)^2\rho_2^2\sum_{\{\epsilon_i\}=\pm}\epsilon_1\epsilon_2\epsilon_3\epsilon\epsilon'\int d\tau d\tau_1 d\tau_2 d\tau_3 d\tau' a(\tau)^3a(\tau_1)^2\theta'(\tau_2)a(\tau_3)^2a(\tau')^3\epsilon^{jmk}q_{1i}p_m k_{1l}\\
	\nonumber&&~~~~~~~~~~~~~~~~~~~~~~~~~~~~~~~\times G_{\varphi;\epsilon_1}(q_1,\tau_1)G_{h;\epsilon_1\epsilon}(q_2,\tau_1,\tau)\partial_\tau G_{\varphi;\epsilon}(q_2,\tau)\\
	\nonumber&&~~~~~~~~~~~~~~~~~~~~~~~~~~~~~~~\times D_{ij;\epsilon_1\epsilon_2}(p,\tau_1,\tau_2)D_{kl;\epsilon_2\epsilon_3}(p,\tau_2,\tau_3)\\
	\nonumber&&~~~~~~~~~~~~~~~~~~~~~~~~~~~~~~~\times G_{\varphi;\epsilon_3}(k_1,\tau_3)G_{h;\epsilon_3\epsilon'}(k_2,\tau_3,\tau')\partial_\tau' G_{\varphi;\epsilon'}(k_2,\tau')\\
	&\equiv&F(\vec{q}_1,\vec{q}_2,\vec{k}_1,\vec{k}_2)~(\vec{q}_1+\vec{q}_2)\cdot(\vec{q}_1\times\vec{k}_1)~,
\end{eqnarray}
where in the second step we used the fact that $D_{ij}\sim \delta_{ij}$ effectively. The explicit form of the propagators is dependent on their IR oscillation frequencies $\mu_h$ and $\mu_Z$, which are related to the field masses $m_h$ and $m_Z$ by
\begin{equation}
	\frac{m_h^2}{H^2}=\mu_h^2+\frac{9}{4}~~~\text{and}~~~\frac{m_Z^2}{H^2}=\mu_Z^2+\frac{1}{4}~.
\end{equation}
Because the inflaton is neutral, the four external lines should be completely symmetrized. Taking into account the symmetry $F(\#1,\#2,\#3,\#4)=F(\#3,\#4,\#1,\#2)$, we find the total 4-point function to be
\begin{eqnarray}\label{RealAmp}
	\nonumber\left\langle\varphi_{\vec{q}_1}\varphi_{\vec{q}_2}\varphi_{\vec{k}_1}\varphi_{\vec{k}_2}\right\rangle'&=&\Bigg\{\frac{1}{2}(\vec{q}_1+\vec{q}_2)\cdot[(\vec{q}_1-\vec{q}_2)\times(\vec{k}_1-\vec{k}_2)]\left[F(\vec{q}_1,\vec{q}_2,\vec{k}_1,\vec{k}_2)-F(\vec{q}_1,\vec{q}_2,\vec{k}_2,\vec{k}_1)\right]\\
	\nonumber&&+\left(\vec{q}_2\leftrightarrow\vec{k}_2\right)\\
	&&+\begin{pmatrix}
	\vec{q}_2 & \vec{k}_1 & \vec{k}_2 \\
	\downarrow & \downarrow & \downarrow \\
	\vec{k}_1 & \vec{k}_2 & \vec{q}_2
	\end{pmatrix}\Bigg\}+\Bigg\{F(\#1,\#2,\#3,\#4)\rightarrow F(\#2,\#1,\#4,\#3)\Bigg\}~,
\end{eqnarray}
where the three lines represent correspondingly $t,u,s$ channel contributions. This is again of the same form of (\ref{flatAmp}) and (\ref{ToyAmp}), hence sharing the $\sin\phi$ dependence for the two planes defined by four momenta. The coefficient function $F$ is dictated by detailed dynamics and can be calculated numerically. To simplify a five-layer time-ordered integral, we evoke the mixed propagator that was introduced in \cite{Chen:2017ryl} and reduce the integral to three layers.

From the definition of the trispectrum $T$ in (\ref{trispec}), we have
\begin{equation}\label{TrispecNorm}
	T^{PT}(\vec{q}_1,\vec{q}_2,\vec{k}_1,\vec{k}_2)=\frac{\dot{\phi}_0^2}{H^4}\frac{\left\langle\varphi_{\vec{q}_1}\varphi_{\vec{q}_2}\varphi_{\vec{k}_1}\varphi_{\vec{k}_2}\right\rangle'}{H^4}\frac{(q_1 q_2 k_1 k_2)^3}{K^3}.
\end{equation}
The trispectrum calculated from the above expression behaves as an odd function of the angle $\phi$. 

We note that the trispectrum induced from one $\theta$ insertion is purely imaginary. This is a notable fact due to P-violation. While the scalar correlation function in position space is manifestly real, its counterpart in momentum space is in general not. Because
\begin{equation}
	\left\langle\prod_j^n\varphi(\vec{x}_j)\right\rangle=\int_{\{\vec{k}_j\}}e^{i\sum_j^n \vec{k}_j\cdot\vec{x}_j}\left\langle\prod_j^n\varphi_{\vec{k}_j}\right\rangle=\left\langle\prod_j^n\varphi(\vec{x}_j)\right\rangle^*
\end{equation}
leads to
\begin{equation}\label{NptPrelation}
	\left\langle\prod_j^n\varphi_{-\vec{k}_j}\right\rangle=\left\langle\prod_j^n\varphi_{\vec{k}_j}\right\rangle^*~.
\end{equation}
For $n<4$, we can use spatial rotations (if there is rotational symmetry) to transform the left-hand side of (\ref{NptPrelation}) back to the original configuration, thereby establishing the reality. However, for $n\geqslant 4$, P-violation leads to $\left\langle\prod_j\varphi_{-\vec{k}_j}\right\rangle \neq \left\langle\prod_j\varphi_{\vec{k}_j}\right\rangle$. This will give rise to the imaginary part of the 4-point correlation function in momentum space. FIG.~\ref{phiDep} shows the imaginary part of the dimensionless trispectrum with respect to $\phi$ for different momentum configurations. The shape dependence on the angle $\phi$ is in accordance with our expectation from intuition. For example, in the middle panel where $k_1=k_2$, we anticipate the behavior $|\Im \tilde{T}^{PT}(\phi=\delta)|=|\Im \tilde{T}^{PT}(\phi=\pi-\delta)|$ because of rotational symmetry.
\begin{figure}[h!]
	\centering
	\nonumber\includegraphics[width=17cm]{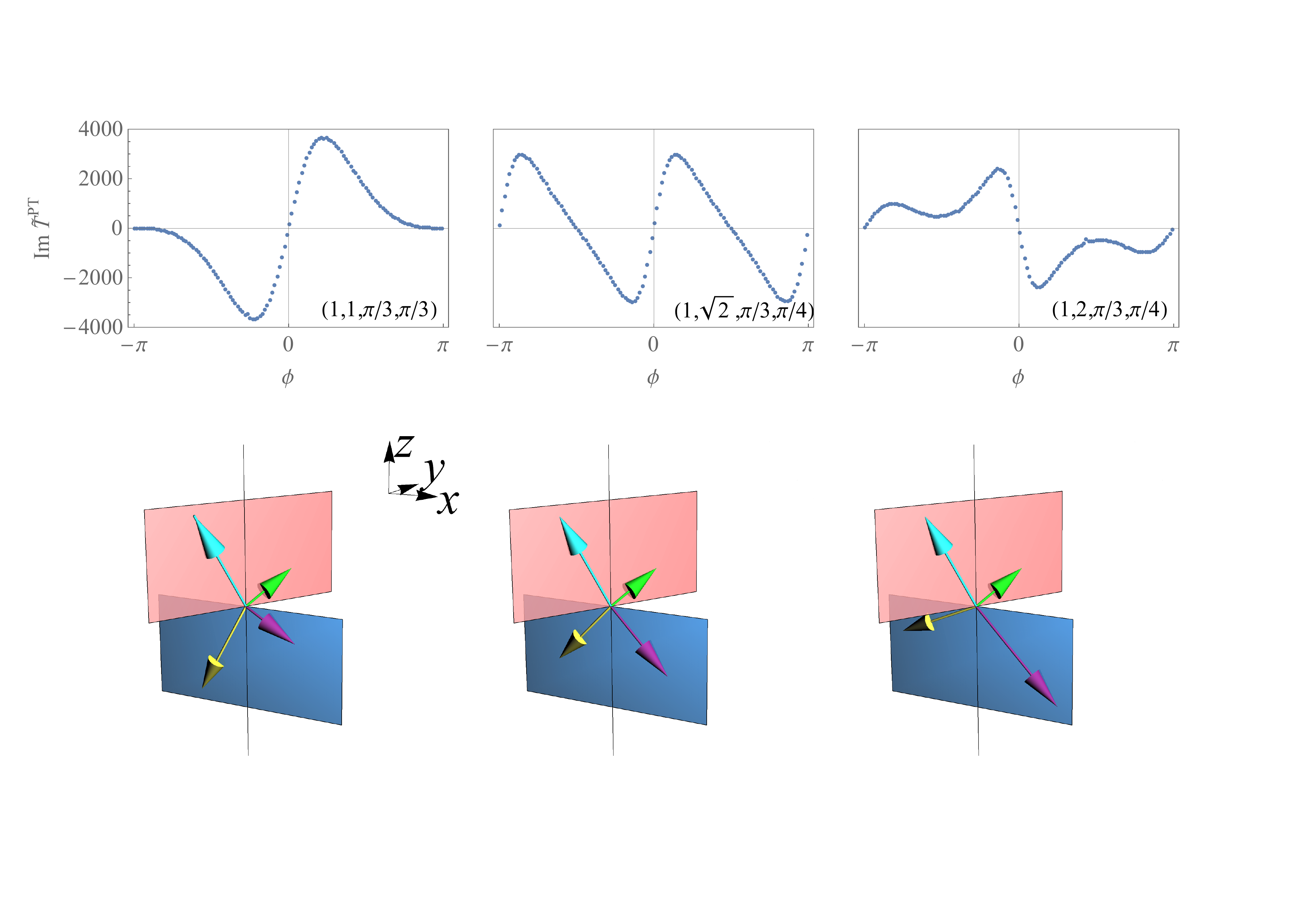}\\
	\caption{The perturbatively computed dimensionless trispectrum divided by couplings $\tilde{T}^{PT}\equiv T^{PT}/c(\frac{\rho_2}{H})^2(\frac{\rho_{1,Z}}{\dot{\phi}_0})^2$ as a function of $\phi$ with different momentum configurations. Left panel: $k_1=q_1=1,p=2,\theta_k=\theta_q=\pi/3$, middle panel: $k_1=1,q_1=\sqrt{2},p=2,\theta_k=\pi/3,\theta_q=\pi/4$, right panel: $k_1=1,q_1=p=2,\theta_k=\pi/3,\theta_q=\pi/4$. The masses are chosen as $\mu_h=0.3$, $\mu_Z=0.2$, which correspond to $m_h=1.53H, m_Z=0.54H$.}\label{phiDep}
\end{figure}
In this light-field case, $\tau_{NL}\sim \mathcal{O}(10)$ if the couplings are chosen as $c\sim 0.1,\rho_2/H,\rho_{1,Z}/\dot{\phi}_0\sim 0.2$. When the mass of the fields increases, the decrease in $\tau_{NL}$ is significant (see Sect.~\ref{LargeMassEFT}).

\subsection{Partially non-perturbative treatment of the $\theta$ term}\label{EOMmethod}

The calculation above is for a single $\theta$-term insertion. This is essentially the leading P- and CP-violating term of a perturbative expansion in terms of $c\ll 1$, where $c$ is defined in (\ref{cdef}). However, it is physically allowed to have $c\sim\order{1}$, where the perturbative expansion in small $c$ is no longer valid. In this case we should treat the $\theta$-term non-perturbatively. The way to keep contributions to all orders in $c$ is to derive the EOM for the gauge field by including the $\theta$-term. The resulting equation is still linear and has an analytical solution. Then we can use the mode function from this equation to compute the trispectrum, which effectively includes contributions with arbitrary number of $\theta$-term insertions. We illustrate this resummation as below.
\begin{equation}
	\begin{gathered}
	\includegraphics[width=3cm]{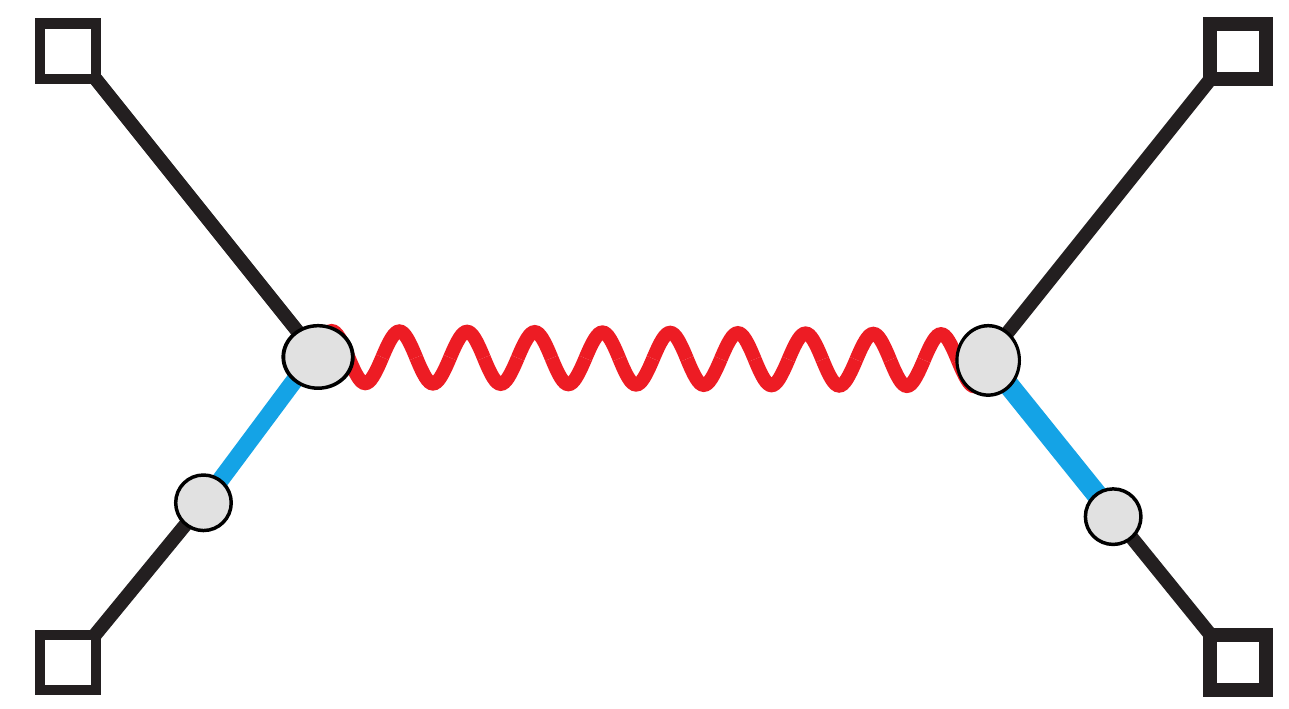}
	\end{gathered}=\begin{gathered}
	\includegraphics[width=3cm]{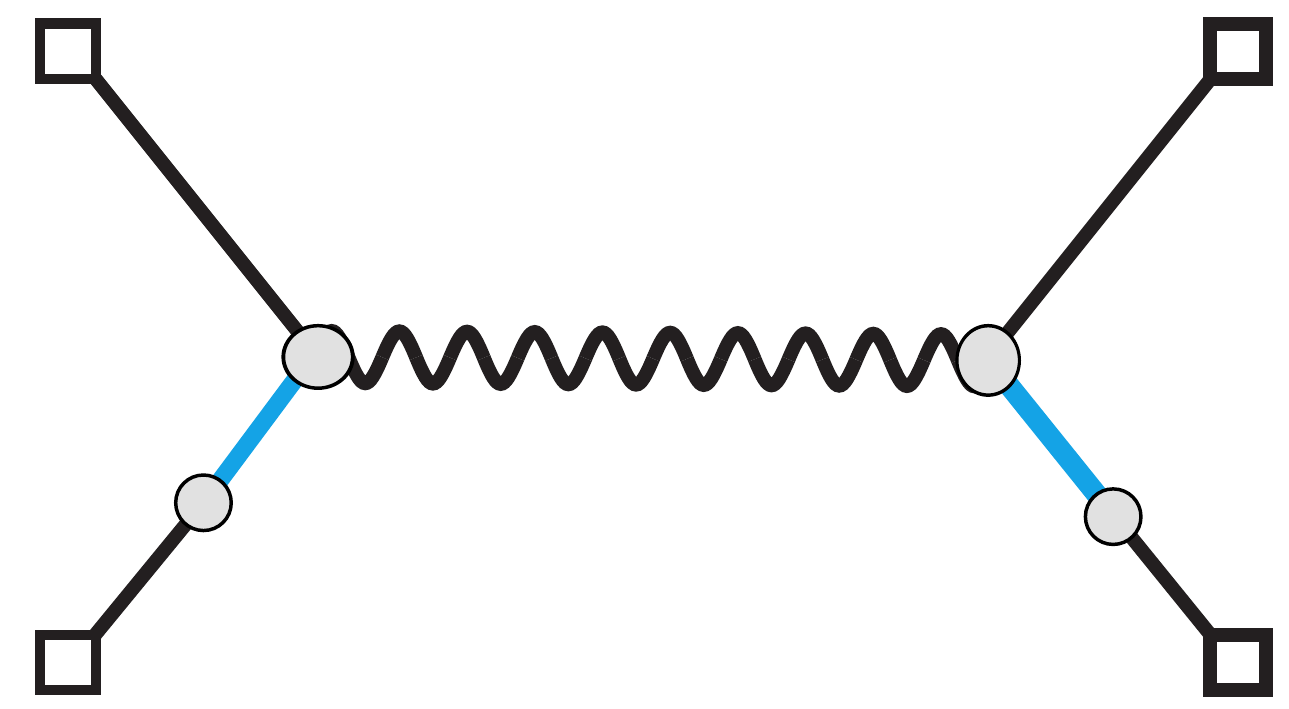}
	\end{gathered}+\begin{gathered}
	\includegraphics[width=3cm]{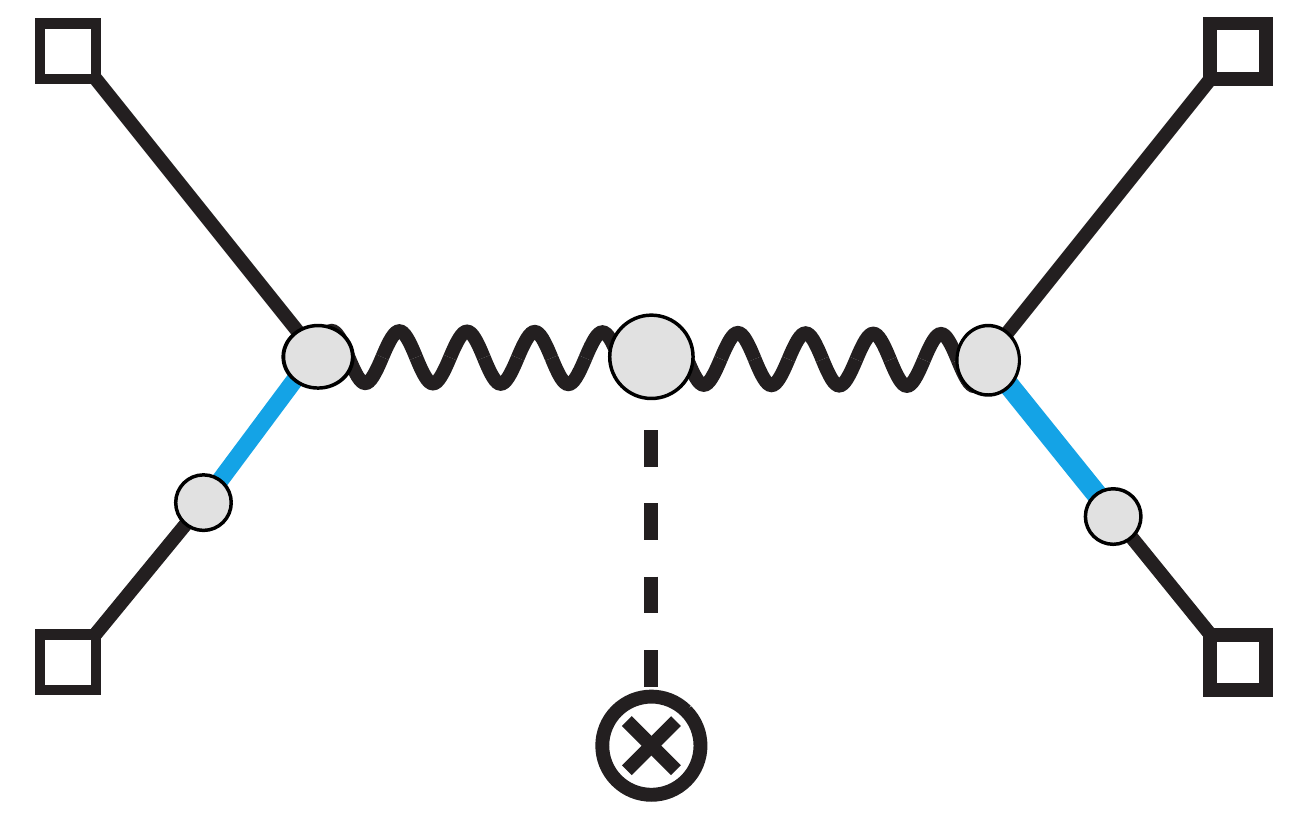}
	\end{gathered}+\begin{gathered}
	\includegraphics[width=3cm]{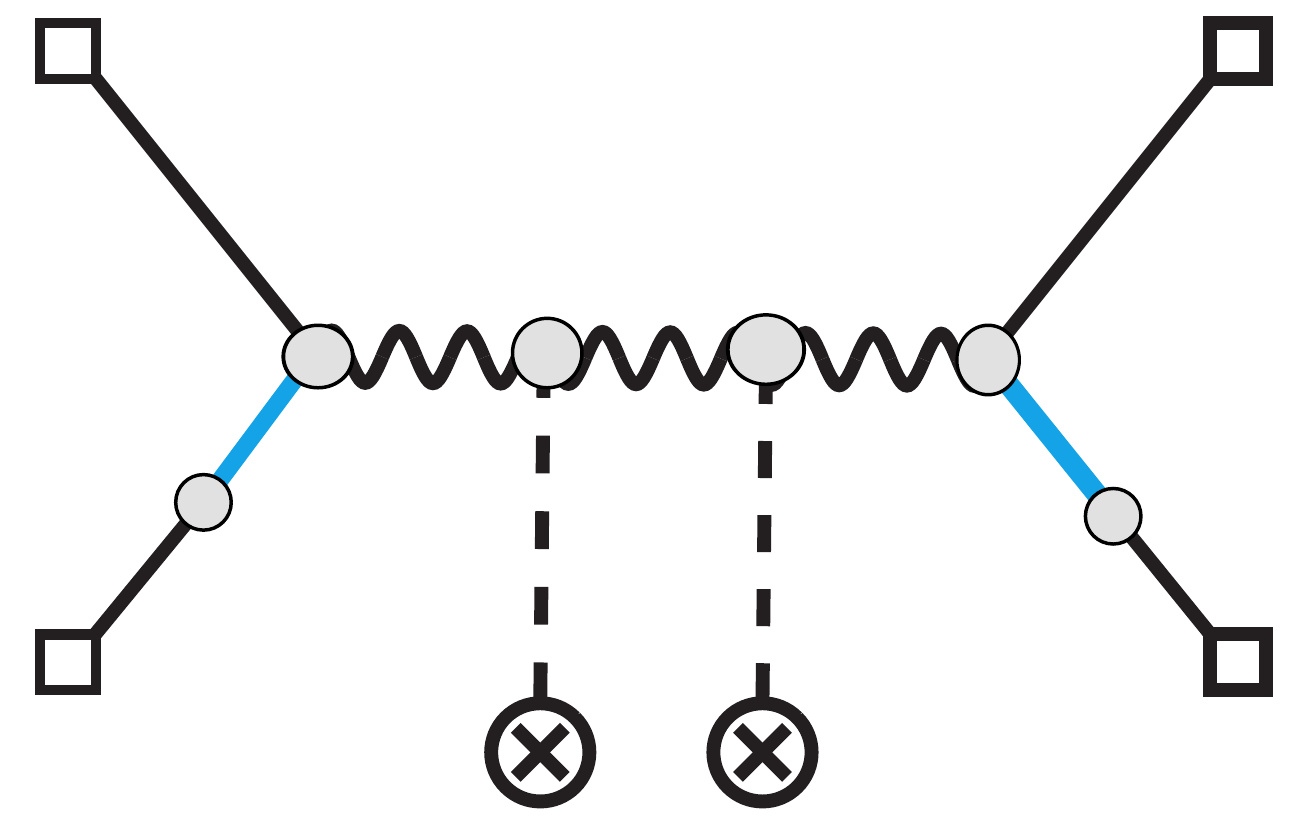}
	\end{gathered}+\cdots~.
\end{equation}

We could have used this method from the very beginning. But a single $\theta$-insertion is still useful because it is in direct analogy to one-body decay in particle physics. Furthermore, we can see explicitly the appearance of the Levi-Civita symbol, both in the original Lagrangian and in the final result. As we shall see below, the Levi-Civita dependence in our non-perturbative treatment is no longer manifest in intermediate steps, but the P-odd angular dependence in the final result still persists.

The EOM of the vector field $Z$ is obtained by varying the quadratic action,
\begin{equation}\label{ZEOMindS}
	\partial_\mu Z^{\mu\sigma}-m_Z^2 a^2 Z^\sigma=-c_0\partial_\rho \theta \epsilon^{\mu\nu\rho\sigma}Z_{\mu\nu}~.
\end{equation}
Notice that the indices are raised using $\eta^{\mu\nu}$, as will be for this whole subsection. In the unitary gauge, $Z$ boson behaves as a Proca field with a second-class constraint found by taking the divergence of (\ref{ZEOMindS}),
\begin{equation}\label{ProcaConstraint}
	\partial_\sigma(a^2 Z^\sigma)=0~,
\end{equation}
which becomes $2HZ_0=\partial_\sigma Z^\sigma$ in dS. Therefore $Z_0$ has no dynamics and must be solved from the dynamics of the longitudinal component. Since we are interested in the effects brought by the $\theta$ term, we neglect the longitudinal dynamics and write the Fourier-space EOM for the transverse components $Z^\bot_i \equiv (\delta_{ij}-\partial_i\partial_j/\partial_k^2)Z_j$ as,
\begin{equation}
	\left[\partial_\tau^2+(p^2+m_Z^2a^2)\right]Z_i^\bot(\vec{p})=2ic_0a\theta'\epsilon_{ijk} p_j Z_k^\bot(\vec{p})~.
\end{equation}
We choose circular polarizations as the basis to diagonalize the EOM, namely $Z_i^\bot=\sum_{\lambda=\pm}\epsilon^\lambda_i(\hat{p})v_p^\lambda(\tau)$ and $\vec{\epsilon}^{~(\pm)}=(\vec{\epsilon}^{~(1)}\pm i\vec{\epsilon}^{~(2)})/\sqrt{2}$, where $\hat{p}$ stands for the unit vector in the $\vec{p}$ direction. With the inflationary background, the EOM reads
\begin{equation}
	v_p^{(\pm)\prime\prime}+\left(p^2\mp\frac{2p c_0\dot{\theta}}{H\tau}+\frac{m_Z^2}{H^2\tau^2}\right)v_p^{(\pm)}=0~.
\end{equation}
The second term in the bracket comes from the CP-violating $\theta$ term and acts as a chemical potential favoring one polarization over the other and thus produces a left-right imbalance. As mentioned above, this is the source of P-violation in the four-scalar final state. The solutions to the EOM are Whittaker functions, which under the Bunch-Davies initial condition become
\begin{equation}
	v_p^{(\pm)}=\frac{1}{\sqrt{2p}}2^{\mp ic}e^{\mp\pi c/2}W(\pm ic,i\mu_Z,2ip\tau)\xrightarrow{\tau\rightarrow-\infty}\frac{1}{\sqrt{2p}}e^{-ip\tau}(-p\tau)^{\pm i c}~,
\end{equation}
From this mode function we obtain four Schwinger-Keldysh propagators,
\begin{subequations}\label{VecProp}
	\begin{eqnarray}
	D_{i_1 i_2}^{+-}(\vec{p},\tau_1,\tau_2)&=&\sum_{\lambda=\pm}\left[\epsilon_{i_1}^\lambda(-\hat{p}) v_p^\lambda(\tau_1)\right]^*\epsilon_{i_2}^\lambda(-\hat{p}) v_p^\lambda(\tau_2)~~~\\
	D_{i_1 i_2}^{-+}(\vec{p},\tau_1,\tau_2)&=&\sum_{\lambda=\pm}\epsilon_{i_1}^\lambda(\hat{p}) v_p^\lambda(\tau_1)\left[\epsilon_{i_2}^\lambda(\hat{p}) v_p^\lambda(\tau_2)\right]^*~~~\\
	D_{i_1 i_2}^{++}(\vec{p},\tau_1,\tau_2)&=&\Theta(\tau_1-\tau_2)D_{i_1 i_2}^{-+}(\vec{p},\tau_1,\tau_2)+\Theta(\tau_2-\tau_1)D_{i_1 i_2}^{+-}(\vec{p},\tau_1,\tau_2)\\
	D_{i_1 i_2}^{--}(\vec{p},\tau_1,\tau_2)&=&\Theta(\tau_1-\tau_2)D_{i_1 i_2}^{+-}(\vec{p},\tau_1,\tau_2)+\Theta(\tau_2-\tau_1)D_{i_1 i_2}^{-+}(\vec{p},\tau_1,\tau_2)~.
	\end{eqnarray}
\end{subequations}
As a consistency check, the propagators are invariant under $SO(2)$ little group transformations $\epsilon_{i_1}^\lambda\rightarrow e^{i\lambda\theta}\epsilon_{i_1}^\lambda$, $\epsilon_{i_2}^{\lambda*}\rightarrow e^{-i\lambda\theta}\epsilon_{i_2}^{\lambda*}$. However, notice that now with P-violation, we have
\begin{equation}
	D_{i_1 i_2}^{+-}(\vec{p},\tau_1,\tau_2)=\Big[D_{i_1 i_2}^{-+}(\vec{p},\tau_1,\tau_2)^*\Big]_{\vec{p}\rightarrow-\vec{p}}\neq D_{i_1 i_2}^{-+}(\vec{p},\tau_1,\tau_2)^*~,
\end{equation}
because the mode functions for two polarizations behave differently and $v_p^{(+)}\neq v_p^{(-)}$. Hence the vector propagator distinguishes two polarizations, which in turn will be imprinted on the final states. 

With the effects of the $\theta$ term non-perturbatively encoded in the propagators in (\ref{VecProp}), we can calculate the trispectrum by a simple exchange diagram. The current insertion vertex requires a contraction between polarization vectors and the corresponding three-momenta. To proceed, we build the polarization vectors through a Gram-Schmidt procedure,
\begin{equation}
	\vec{\epsilon}^{~(\pm)}=\frac{1}{\sqrt{2(1-(\hat{n}\cdot\hat{p})^2)}}\left[(\hat{n}-(\hat{n}\cdot\hat{p})\hat{p})\pm i(\hat{p}\times\hat{n})\right]~,
\end{equation}
where $\hat{n}$ is a random-directional unit vector different from $\hat{p}$. The momentum contraction involves the following expression:
\begin{equation}
	\left(\vec{q}_1\cdot\vec{\epsilon}^{~(\pm)}(\hat{p})\right)\left(\vec{k}_1\cdot\vec{\epsilon}^{~(\pm)}(\hat{p})\right)^*=\left(\vec{q}_1\cdot\vec{\epsilon}^{~(\pm)}(-\hat{p})\right)^*\left(\vec{k}_1\cdot\vec{\epsilon}^{~(\pm)}(-\hat{p})\right)~.
\end{equation}
Since we have checked the little group invariance of the vector propagators, we can choose whatever $\hat{n}$ that simplifies calculation without affecting the final result. Setting $\hat{n}=\hat{q}_1$ gives
\begin{equation}
	\left(\vec{q}_1\cdot\vec{\epsilon}^{~(\pm)}(\hat{p})\right)\left(\vec{k}_1\cdot\vec{\epsilon}^{~(\pm)}(\hat{p})\right)^*=\frac{1}{2}\left[\vec{q}_1\cdot\vec{k}_1-(\vec{q}_1\cdot\hat{p})(\vec{k}_1\cdot\hat{p})\mp i\hat{p}\cdot(\vec{q}_1\times\vec{k}_1)\right]~.
\end{equation}
The P-odd pattern appears again in the last term. If the $\theta$ term were absent, the mode functions for two polarizations would be the same and would lead to a cancellation of this pattern, leaving a real trispectrum without P-violation.

The 4-point function is now computed easily as a single tree-level exchange diagram,
\begin{eqnarray}
\nonumber\left\langle\varphi_{\vec{q}_1}\varphi_{\vec{q}_2}\varphi_{\vec{k}_1}\varphi_{\vec{k}_2}\right\rangle'&=&-\left(\frac{\rho_{1,Z}}{\dot{\phi}_0}\right)^2\rho_2^2\sum_{\{\epsilon_i\}=\pm}\epsilon_1\epsilon_2\epsilon\epsilon'\int d\tau d\tau_1 d\tau_2 d\tau' a(\tau)^3a(\tau_1)^2a(\tau_2)^2a(\tau')^3\\
\nonumber&&~~~~~~~~~~~~~~~~~~~~~~~~~~~~~~~\times G_{\varphi;\epsilon_1}(q_1,\tau_1)G_{h;\epsilon_1\epsilon}(q_2,\tau_1,\tau)\partial_\tau G_{\varphi;\epsilon}(q_2,\tau)\\
\nonumber&&~~~~~~~~~~~~~~~~~~~~~~~~~~~~~~~\times q_{1i_1}D_{i_1 i_2}^{\epsilon_1\epsilon_2}(\vec{p},\tau_1,\tau_2)k_{1i_2}\\
\nonumber&&~~~~~~~~~~~~~~~~~~~~~~~~~~~~~~~\times G_{\varphi;\epsilon_2}(k_1,\tau_2)G_{h;\epsilon_2\epsilon'}(k_2,\tau_2,\tau')\partial_\tau' G_{\varphi;\epsilon'}(k_2,\tau')\\
&&+(\text{23 perms})~.
\end{eqnarray}
Afterwards, calculations are standard and the final trispectrum normalized according to (\ref{TrispecNorm}) is shown in FIG.~\ref{phiDepNonpert}. As is clear from the figure, in the presence of the $\theta$ term, the trispectrum induced by the transverse components of the vector boson develops an imaginary part that behaves as an odd function of $\phi$ around the planar configuration. In contrast, the real part of the trispectrum is an even function in $\phi$. The physical explanation for this is very clear. Planar momentum configurations are even under P transformation, and therefore cannot possess an imaginary part\footnote{Notice that this is true only when spatial rotational symmetry is preserved. If there exists a special direction, planar configurations can also have nonzero imaginary parts (see \cite{Shiraishi:2016mok}). Here we deem manifest rotational symmetry as a more natural choice and will only consider this case hereafter.}, for the same reason as why in Chemistry, planar molecules generally have no enantiomers. To check the consistency, a comparison with the leading-order perturbation theory results in the previous subsection is shown in FIG.~\ref{PertVSNonpert}. Clearly, in the perturbative regime, these two methods agree with each other very well. In the partially non-perturbative regime, $e.g.$, $c=0.6$, $\Im(c\tilde{T}^{PT})$ mismatches $\Im\tilde{T}^{EOM}$ by a numerical factor.

\begin{figure}[h!]
	\includegraphics[width=17cm]{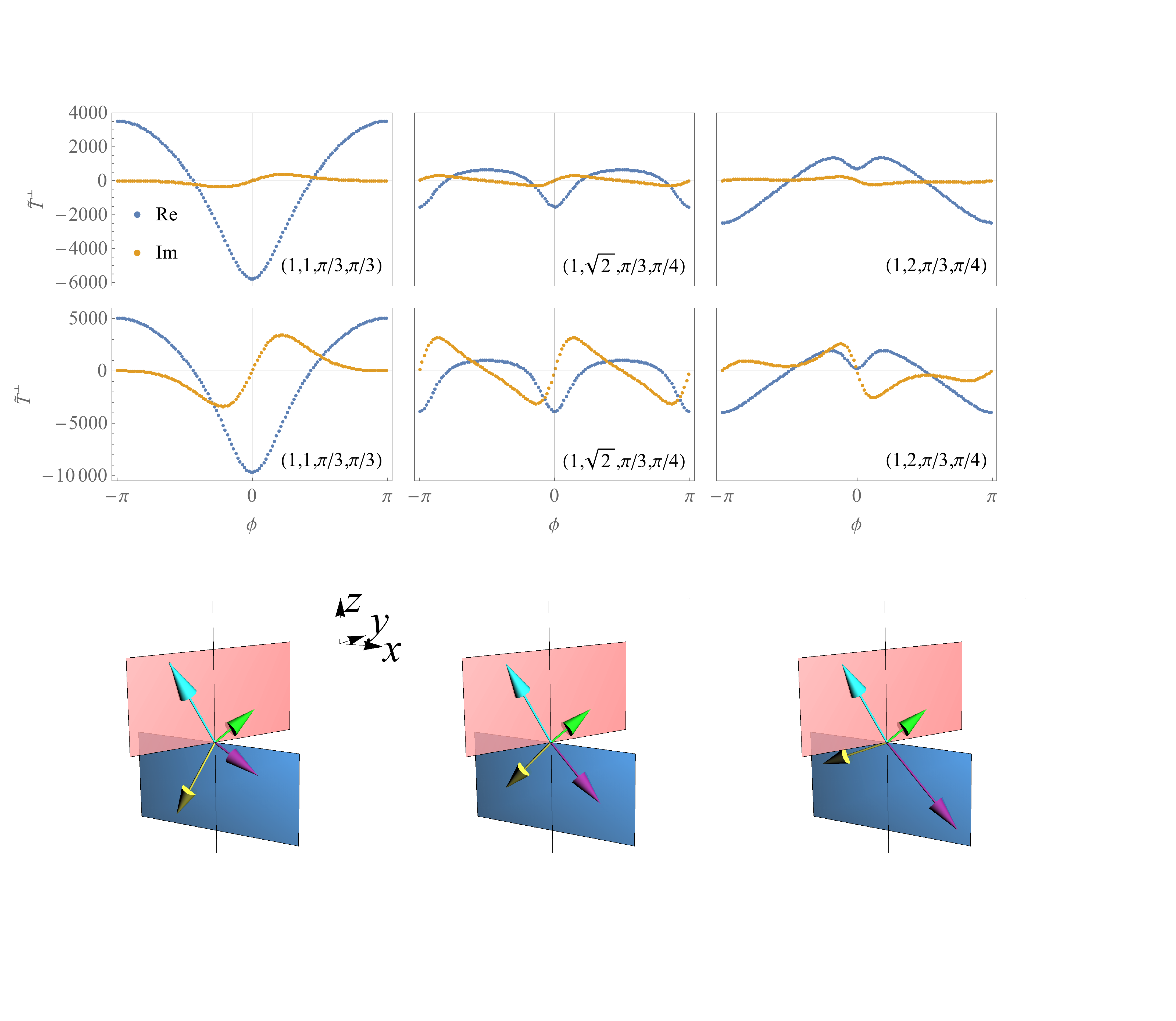}
	\caption{The dimensionless trispectrum divided by couplings $\tilde{T}^\bot\equiv T^\bot/(\frac{\rho_2}{H})^2(\frac{\rho_{1,Z}}{\dot{\phi}_0})^2$ as a function of $\phi$ with different momentum configuration. Left panel: $k_1=q_1=1,p=2,\theta_k=\theta_q=\pi/3$, middle panel: $k_1=1,q_1=\sqrt{2},p=2,\theta_k=\pi/3,\theta_q=\pi/4$, right panel: $k_1=1,q_1=p=2,\theta_k=\pi/3,\theta_q=\pi/4$. Here we have taken $c=0.1$ (perturbative in $c$) for the first line and $c=0.6$ (non-perturbative in $c$, marginal in loops) for the second line. The masses are chosen as $\mu_h=0.3$, $\mu_Z=0.2$, which correspond to $m_h=1.53H, m_Z=0.54H$. In these cases, $\Im\tau_{NL}\sim \mathcal{O}(10^{2})$ if the couplings are all near $0.2$.}\label{phiDepNonpert}
\end{figure}

\begin{figure}[h!]
	\centering
	\includegraphics[width=15cm]{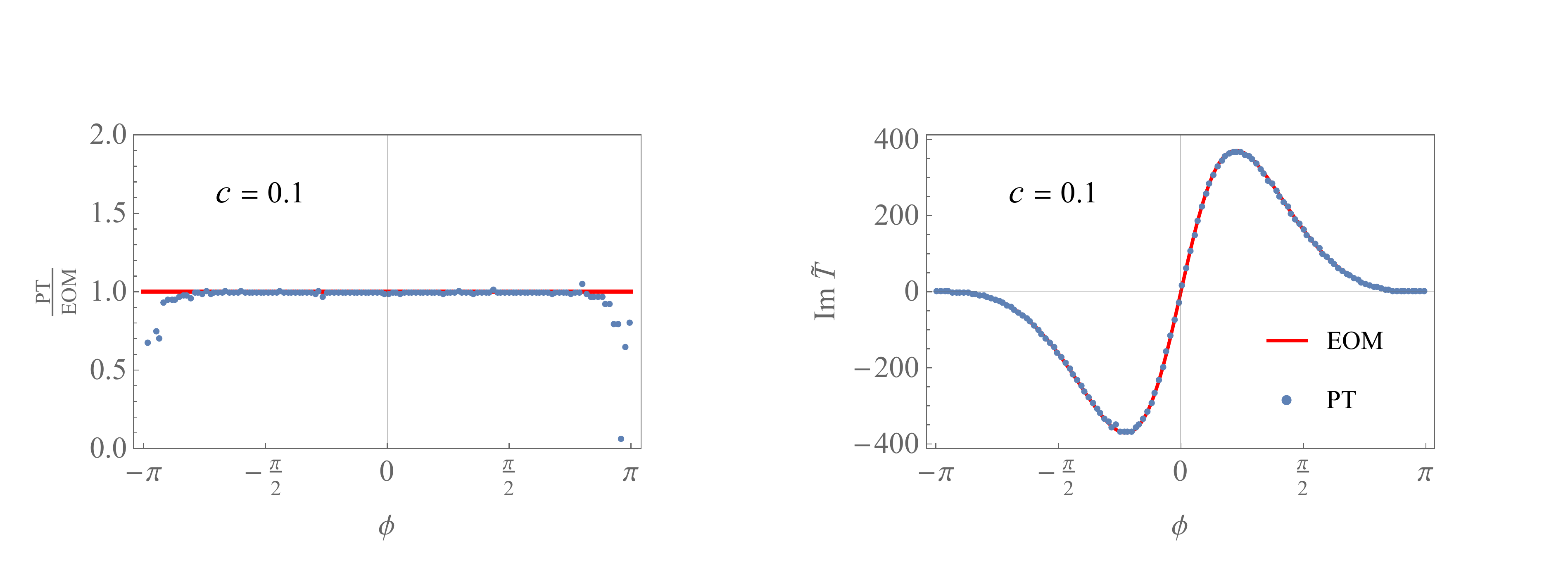}\\
	\includegraphics[width=15cm]{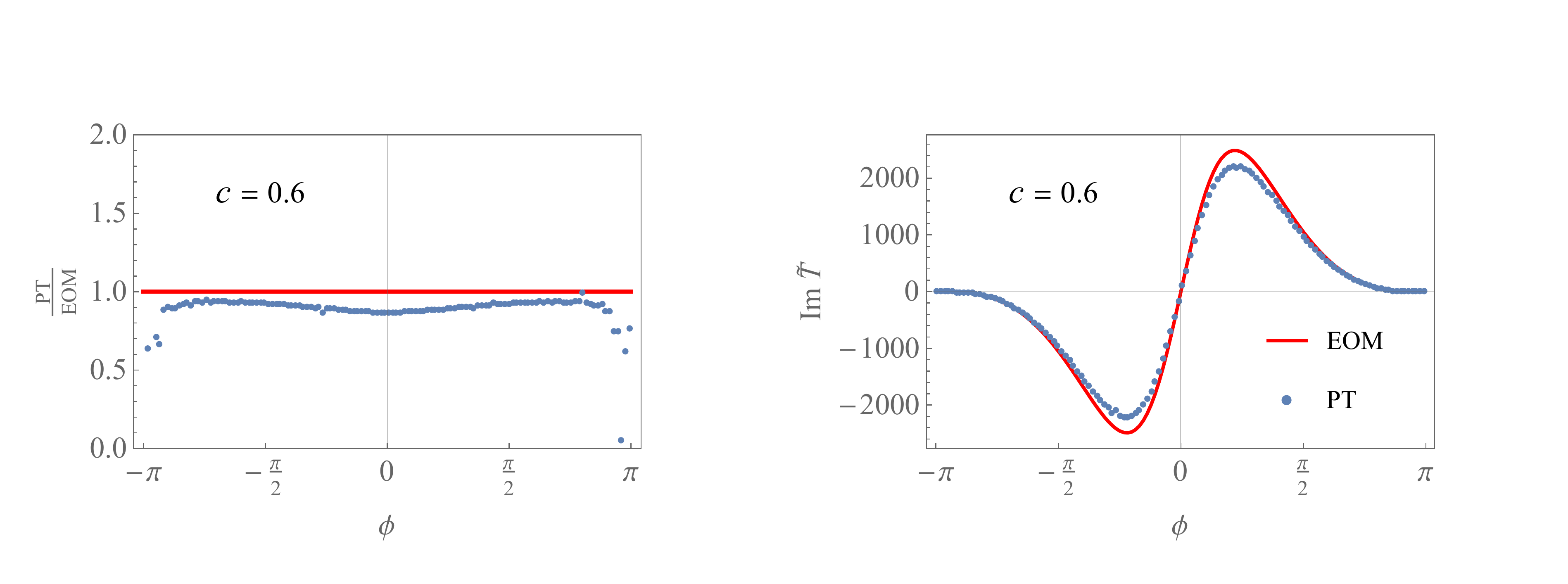}
	\caption{A numerical check on the consistency. In the left panel, we plot the ratio of perturbative results obtained in Sect.~\ref{1PTmethod} against the non-perturbative results obtained using EOM in Sect.~\ref{EOMmethod}. In the right panel, we show the imaginary part of the trispectrum on the same plot for a direct comparison. The first row corresponds to $c=0.1$ while the second row corresponds to $c=0.6$. The parameters are chosen as $k_1=q_1=1,p=2,\theta_k=\theta_q=\pi/3$ and $\mu_h=0.3,\mu_Z=0.2$, which correspond to $m_h=1.53H, m_Z=0.54H$. For $c=0.1$, within most regions, the error is acceptable by $\lesssim c=10\%$, validating perturbation theory. Near the ends,  numerical uncertainties overcome the systematic deviation predicted by perturbation theory, since $\Im \tilde{T}\sim 0$. For $c=0.6$, the two methods approximately mismatch by a numerical factor.}\label{PertVSNonpert}
\end{figure}

When chemical potential is large, namely $c\gtrsim 1$, the production rate of gauge boson is dramatically amplified. This can be seen from the IR expansion of the gauge field mode function:
\begin{equation}\label{IRexpansion}
	v_p^{(\pm)}\xrightarrow{\tau\rightarrow 0}\alpha_\pm \frac{C}{\sqrt{2\mu_Z}}(-\tau)^{\frac{1}{2}+i\mu_Z}+\beta_\pm \frac{C^*}{\sqrt{2\mu_Z}}(-\tau)^{\frac{1}{2}-i\mu_Z}~,
\end{equation}
where $C=e^{i(\mu_Z\ln (2p)-\pi/4)}$ is a pure phase and
\begin{equation}
	\alpha_\pm=2^{\mp i c} e^{-\frac{1}{2} \pi(\pm c-\mu_Z)} \frac{\sqrt{2\mu_Z }~\Gamma (-2 i \mu_Z )}{\Gamma \left(\frac{1}{2}-i \mu_Z\mp i
		c\right)},~\beta_\pm=-i 2^{\mp i c} e^{-\frac{1}{2} \pi  (\pm c+\mu_Z )} \frac{\sqrt{2\mu_Z }~\Gamma (2 i \mu_Z )}{\Gamma \left(\frac{1}{2}+i \mu_Z\mp i c\right)}
\end{equation}
are the Bogolyubov coefficients. The particle number density in the momentum space is given by
\begin{equation}
	\langle n_{\vec{p}}^{(\pm)}\rangle'=|\beta_\pm|^2=\frac{1}{r_\pm e^{2\pi\mu_Z}-1},~~~\text{with}~~~ r_\pm=\frac{\cosh[\pi(\pm c+\mu_Z)]}{\cosh[\pi(\pm c-\mu_Z)]}~.
\end{equation}
For a large positive $c$, $r_\pm\rightarrow e^{\pm 2\pi\mu_Z}$, leading to an exponentially enhanced production of negatively polarized gauge field particles, $i.e.$, $\langle n_{\vec{p}}^{(-)}\rangle'\propto e^{2\pi c}$. This exponential growth in particle number density could threaten the inflation background. To check this, we compute the energy density of the produced gauge bosons,
	\begin{equation}
	\langle T_{\mu\nu}(Z)\rangle=-\frac{2}{\sqrt{-g}}\left\langle\frac{\delta S_2[Z]}{\delta g^{\mu\nu}}\right\rangle=\left\langle Z_{\mu\rho}Z_\nu^{~\rho}-\frac{1}{4}g_{\mu\nu}Z_{\rho\sigma}Z^{\rho\sigma}+m_Z^2 Z_\mu Z_\nu-\frac{1}{2}g_{\mu\nu}m_Z^2 Z_\rho Z^\rho\right\rangle~.
	\end{equation}
	Interestingly, the $\theta$ term does not contribute to $T_{\mu\nu}(Z)$ at all because of its ignorance to the geometry of spacetime. The physical energy density is given as $\varepsilon_Z=\langle T_{tt}\rangle=a^{-2}\langle T_{\tau\tau}\rangle$. Considering only the amplified transverse modes and using the mode expansion, we obtain the usual expression for vacuum energy contributed by transverse modes of $Z$,
	\begin{eqnarray}
	\nonumber\varepsilon_Z^\bot&=&a^{-4}\left\langle\frac{1}{2}(\partial_\tau Z_i^\bot)^2+\frac{1}{4}(\partial_i Z_j^\bot-\partial_j Z_i^\bot)^2+\frac{1}{2}m_Z^2 a^2 Z_i^{\bot 2}\right\rangle\\
	&=&\sum_{\lambda=\pm}\int_{\vec{p}}\frac{1}{2 a^4}\left(|v_p^{\lambda\prime}|^2+(p^2+m_Z^2 a^2)|v_p^\lambda|^2\right)~.
	\end{eqnarray}
	The momentum integral is quartically divergent in the UV, as is in flat spacetime. This formally infinite contribution to the energy density by the vacuum fluctuations is always present and we assume it is canceled by a shift in the height of the inflaton potential. Thus we only need to care about the contribution by real particle production. We cut off the momentum integral at horizon scale $p<-\tau^{-1}$ and use the IR expansion (\ref{IRexpansion}) to obtain
	\begin{equation}
	\varepsilon_Z^\bot(\tau)\approx\sum_{\lambda=\pm}\int_{|\vec{p}_{\text{ph}}|<H}\left(m_Z+\frac{p_{\text{ph}}^2}{2m_Z}\right)\left(|\beta_\lambda|^2+\frac{1}{2}\right)~.
	\end{equation}
	As a result, the energy density in the gauge field sector is approximately given by the rest energy and non-relativistic kinetic energy of the produced real particles. Notice that the $1/2$ term in the last bracket is the remaining (finite) vacuum energy in the IR. 
	 Since now $\langle n_{\vec{p}}^{(-)}\rangle'=|\beta_-|^2\propto e^{2\pi c}$, the constraint on energy density becomes
	\begin{equation}\label{energybound}
	\varepsilon_Z^\bot(\tau)\sim H^3 m_Z e^{2\pi c}<-M_p^2\dot{H}= \frac{1}{2}\dot{\phi}_0^2 ~~\Rightarrow~~ c< \frac{1}{2\pi}\ln\frac{\dot{\phi}_0^2/H^4}{2m_Z/H}~.
	\end{equation}
We see that with $\dot{\phi}_0^2/H^4\sim 3600$, for $\mu_Z\sim 0.2$, or $m_Z/H\sim 0.54$, the chemical potential should satisfy $c\lesssim1.29$ to keep the stability of inflation background. However, even when (\ref{energybound}) is satisfied, the calculation in this subsection may still be inapplicable due to the break-down of loop expansion\footnote{Loop expansion is only a way of organizing calculation. What really makes perturbation theory valid is the smallness of $all$ coupling constants.}. If the particles running the loop include the highly amplified transverse modes of $Z$, the loop corrections may exceed the tree-level propagator (\ref{VecProp}) and perturbation theory becomes invalid. In the gauge-scalar sector, there are interactions between Higgs and $Z$ boson that contribute to the self-energy of the $Z$ boson at one-loop level, which goes like
	\begin{equation}
	\begin{gathered}
	\includegraphics[width=4cm]{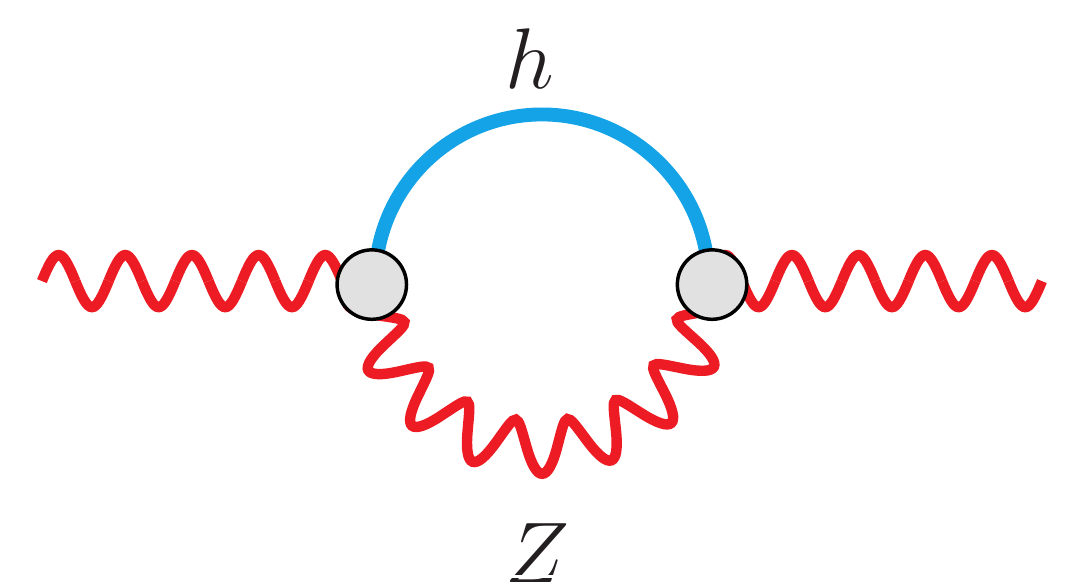}
	\end{gathered}\sim \frac{g^2}{(4\pi)^2}\times\frac{m_Z^2}{m_h^2}e^{4\pi c}\times\left(\begin{gathered}
	\includegraphics[width=2.5cm]{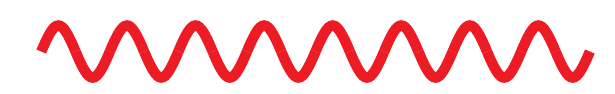}
	\end{gathered}\right)~.
	\end{equation}
	Here the coupling $g m_Z/\cos\theta_W\sim g m_Z$, the loop factor $1/(4\pi)^2$ and the EFT leading order $G_h\sim 1/m_h^2$ have been taken into account. To ensure the validity of loop expansion, we require
	\begin{equation}\label{PTbound}
	\frac{g^2}{(4\pi)^2}\times\frac{m_Z^2}{m_h^2}e^{4\pi c}\lesssim 1~\Rightarrow~c \lesssim\frac{1}{2\pi}\ln\frac{4\pi m_h}{g m_Z}~.
	\end{equation}
	For the parameter choice in FIG.~\ref{phiDepNonpert}, with $g=0.55$, the constraint gives $c\lesssim 0.66$, which suggests a marginally perturbative choice $c=0.6$. We comment that for an SM setup, the $\theta$ coefficient before $Z\tilde{Z}$ is the $SU(2)_L$ axion and hence will also appear before $W^\pm\tilde{W}^\mp$. Due to the non-linear coupling between $W^\pm$ and $Z$, the constraint from loop expansion will be even tighter. Whatsoever, the CP-breaking pattern in the trispectrum that we are concerned with still persists.
	
We stress that the bound (\ref{PTbound}) is from the validity of the perturbative loop expansion. The parameter space violating (\ref{PTbound}) can still be physically conceivable, as long as the back-reaction to the energy budget of the inflationary universe remains small, $i.e.$, when (\ref{energybound}) is satisfied. In such a case, the failure of loop expansion is related to the fact that the copious production of real particles forms a classical many-body system. Thus the loop expansion based on perturbative EFT around an approximate vacuum is no longer sufficient to capture the physics. Then the calculations in this subsection are no longer valid. We leave this possibility for future studies.

In addition to the above constraints from the consistency of our model itself, we need to check the observational constraints on the bispectrum. Since all the tree-level diagrams involving the vector field $Z$ contain the scalar-vector mixing vertex $\partial_\mu\varphi Z^\mu$, only the longitudinal component of $Z$ contributes. Therefore at tree-level, there is no exponentially enhanced diagram. At one-loop level, there do exist such potentially dangerous diagrams. We show the three leading diagrams in FIG.~\ref{bispectrumLoopConstr}. For $c=0.66$ and $\mu_Z=0.2$, the exponential factor $|\alpha_-|^4\sim 427$ tends to cancel the loop factor $\frac{1}{(4\pi)^2}\sim 0.0063$, leaving a tree-level-size result suppressed by couplings and masses. For example,
\begin{equation}
	\begin{gathered}
	\includegraphics[width=3cm]{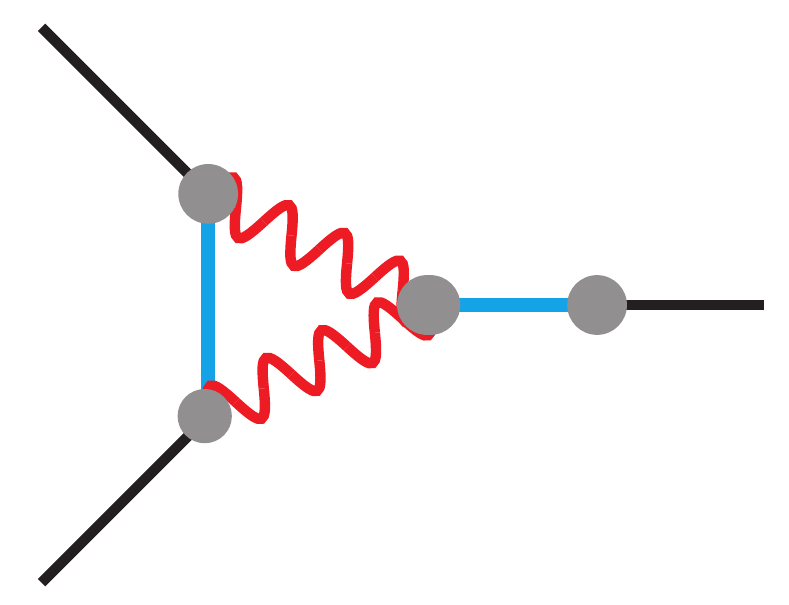}
	\end{gathered}\Rightarrow f_{NL}\sim \frac{\dot{\phi}_0}{H^2}\times|\alpha_-|^4\frac{1}{(4\pi)^2}\left(\frac{\rho_{1,Z}}{\dot{\phi}_0}\right)^2 \frac{g m_Z}{H}\times\frac{\rho_2}{H}\times \frac{H^4}{m_h^4}\approx 4.2
\end{equation}
for $c=0.66,\mu_h=0.3,\mu_Z=0.2,g=0.55,\frac{\rho_{1,Z}}{\dot{\phi}_0}=\frac{\rho_2}{H}=0.2$. Thus the final bispectrum is roughly $f_{NL}\lesssim\mathcal{O}(1)$ under the loop expansion bound (\ref{PTbound}). This still satisfies the current constraints of Planck 2018 \cite{Akrami:2019izv}, which gives $f_{NL}^{\text{local}}=-0.9\pm 5.1, f_{NL}^{\text{equil}}=-26\pm 47, f_{NL}^{\text{ortho}}=-38\pm 24,~(68\% \text{CL})$. For $c>0.66$, the system becomes fully non-perturbative and the naive estimations become invalid.

\begin{figure}[h!]
	\includegraphics[width=13cm]{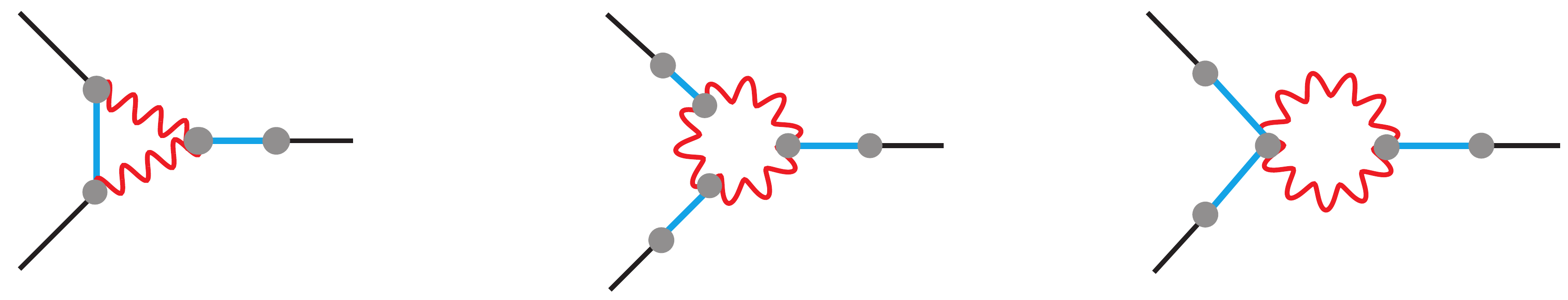}
	\caption{The three leading diagrams with enhancement that contribute to the bispectrum at one-loop level. Their contributions to $f_{NL}$ are estimated to be 4.2, 3.3 and 0.6 for $c=0.66,\mu_h=0.3,\mu_Z=0.2,g=0.55$, with all other couplings near 0.2.}\label{bispectrumLoopConstr}
\end{figure}

Finally, we comment that the real part of the trispectrum is also important for inferring the mass of the vector field $Z$. If we observe the oscillations in the collapsed limit of the real part and the CP-odd imaginary part at the same time, together they can provide better constraints on the model parameters.


\subsection{A large mass EFT?}\label{LargeMassEFT}
In this section, we study the large mass EFT of our model. We will show that in exact dS spacetime, the P- and CP-violating signals cannot be seen at any order in the large mass expansion, therefore demonstrating the importance of on-shell real particle production. The large mass EFT is usually applicable when the massive fields mediating the interactions are heavy compared to $H$. In our study, we also require $c\ll \max\{1,\mu_Z\}$ to suppress the on-shell particle production, focusing on the single field description of the off-shell contributions of extra fields in dS.  For a heavy Higgs, integrating it out yields a change of inflaton sound speed at quadratic level due to the two-point mixing. Furthermore, the original current-potential interaction $\Delta\mathcal{L}_1$ becomes schematically
\begin{equation}
\Delta\mathcal{L}_1=\frac{\rho_{1,Z}}{\dot{\phi}_0}h\partial_\mu \varphi Z^\mu\rightarrow \frac{\rho_{1,Z}}{\dot{\phi}_0}\rho_2\partial_\mu \varphi Z^\mu\frac{1}{\Box-m_h^2}\dot{\varphi}\equiv J_\mu Z^\mu~.
\end{equation}
If the mass of the $Z$ boson is also large, we can integrate it out as well, yielding a current-current interaction from $\Delta\mathcal{L}_3$:
\begin{equation}
\Delta\mathcal{L}_3=-\frac{c_0\theta(t)}{4}Z_{\mu\nu}Z_{\rho\sigma}\mathcal{E}^{\mu\nu\rho\sigma}\rightarrow\Delta\mathcal{L}^\text{EFT}=-c_0\theta(t)\partial_\mu\left[\left(\frac{1}{\Box-m_Z^2}\right)_{\nu\alpha}J^\alpha\right]\partial_\rho\left[\left(\frac{1}{\Box-m_Z^2}\right)_{\sigma\beta}J^\beta\right]\mathcal{E}^{\mu\nu\rho\sigma}~.
\end{equation}
Thus we obtain the EFT Lagrangian by expanding the non-local propagators into an infinite gradient series,
\begin{eqnarray}\label{EFTExpansion}
\nonumber\Delta\mathcal{L}^\text{EFT}&=&-\left(\frac{\rho_{1,Z}}{\dot{\phi}_0}\rho_2\right)^2\frac{c_0\theta(t)}{ m_Z^4 m_h^4}\\
&&\times\mathcal{E}^{\mu\nu\rho\sigma}\sum_{m,n,p,q}\partial_\mu\left[\left(\frac{\Box}{m_Z^2}\right)^m\left(\partial_\nu \varphi\left(\frac{\Box}{m_h^2}\right)^n\dot{\varphi}\right)\right]\partial_\rho\left[\left(\frac{\Box}{m_Z^2}\right)^p\left(\partial_\sigma \varphi\left(\frac{\Box}{m_h^2}\right)^q\dot{\varphi}\right)\right]~.
\end{eqnarray}
Here we used $\eta_{\mu\nu}$ to replace the polarization sum since in the Feynman-diagram calculation only $\delta_{ij}$ coming from the transverse components contributes. Using the antisymmetry of $\epsilon^{\mu\nu\rho\sigma}$, it is easy to see that the first term of (\ref{EFTExpansion}) with $m=n=p=q=0$ vanishes. For simplicity, assuming $m_h\gg m_Z$, we obtain the leading-order (LO) and next-to-leading order (NLO) EFT operators
\begin{subequations}\label{EFTLOandNLO}
	\begin{eqnarray}
	\Delta\mathcal{L}^\text{LO}&=&-\left(\frac{\rho_{1,Z}}{\dot{\phi}_0}\rho_2\right)^2\frac{c_0\theta(t)}{ m_Z^6 m_h^4}\mathcal{E}^{\mu\nu\rho\sigma}2\partial_\mu\left[\Box\left(\partial_\nu \varphi\dot{\varphi}\right)\right]\partial_\rho\left[\partial_\sigma \varphi\dot{\varphi}\right]\\
	\Delta\mathcal{L}^\text{NLO}&=&-\left(\frac{\rho_{1,Z}}{\dot{\phi}_0}\rho_2\right)^2\frac{c_0\theta(t)}{ m_Z^8 m_h^4}\mathcal{E}^{\mu\nu\rho\sigma}\partial_\mu\left[\Box\left(\partial_\nu \varphi\dot{\varphi}\right)\right]\partial_\rho\left[\Box\left(\partial_\sigma \varphi\dot{\varphi}\right)\right]~.
	\end{eqnarray}
\end{subequations}

Consider the expansion of (\ref{EFTLOandNLO}) in a dS background. We first perform an integrate-by-parts (IBPs) to convert the derivative onto $\theta$.
Then by using the EOM, $\varphi''+2a H\varphi'-\partial_i^2\varphi=0$ and doing some IBPs, we obtain the on-shell effective operators
\begin{subequations}
	\begin{eqnarray}
	\sqrt{-g}\Delta\mathcal{L}^\text{LO}&=&\left(\frac{\rho_{1,Z}}{\dot{\phi}_0}\rho_2\right)^2\frac{c_0\dot\theta\epsilon^{ijk}}{ m_Z^6 m_h^4}a^{-3}\left(-4\varphi\partial_i\varphi'\partial_j\partial_l\varphi  \partial_k\partial_l\varphi'\right)\\
	\nonumber\sqrt{-g}\Delta\mathcal{L}^\text{NLO}&=&\left(\frac{\rho_{1,Z}}{\dot{\phi}_0}\rho_2\right)^2\frac{c_0\dot\theta\epsilon^{ijk}}{ m_Z^6 m_h^4}\Bigg\{\frac{H^2}{m_Z^2}a^{-3}\left(-8\varphi\partial_i\varphi'\partial_j\partial_l\varphi\partial_k\partial_l\varphi'\right)\\
	\nonumber&&~~~~~~~~~~~~~~~~~~~~~~~~~~~+\frac{H}{m_Z^2}a^{-4}\left[4  \varphi   \partial_i\partial_n\varphi  \
	\partial_j\partial_m\partial_n\varphi  \partial_k\partial_m\varphi'
	+ 4 \varphi    \partial_i\partial_n\varphi  \
	\partial_j\partial_m\partial_m\varphi  \partial_k\partial_n\varphi '\right]\\
	\nonumber&&~~~~~~~~~~~~~~~~~~~~~~~~~~~+\frac{1}{m_Z^2}a^{-5}\left[-8\partial_m\partial_m\varphi\partial_i\varphi'\partial_j\partial_n\varphi\partial_k\partial_n\varphi'+4\partial_n\varphi'\partial_i\partial_n\varphi\partial_j\partial_m\varphi\partial_k\partial_m\varphi'\right]\Bigg\}~.\\
	&&~~~~~~~~~~~~~~~~~~~~~~~~~~~
	\end{eqnarray}
\end{subequations}
The LO term and the NLO last term survive the flat spacetime limit $H\rightarrow 0$, $a\rightarrow 1$ and are thus present when the spacetime is not expanding. The terms proportional to powers of $H$ are due to the curved spacetime background and can be constructed independently by trading the derivatives on $\varphi$ with those on the spacetime metric. Note that these terms are not total derivatives and naively should contribute to the observables. However, as we shall see in the following, they do not show up in the 4-point function on the dS boundary.

As an explicit example, let us compute the effect of the LO term even before momentum permutation.
\begin{eqnarray}\label{LOnull}
\nonumber\left\langle\varphi_{\vec{q}_1}\varphi_{\vec{q}_2}\varphi_{\vec{k}_1}\varphi_{\vec{k}_2}\right\rangle'_{\text{LO}}&=&4\left(\frac{\rho_{1,Z}}{\dot{\phi}_0}\rho_2\right)^2\frac{c_0\dot\theta}{ m_Z^6 m_h^4}\vec{q}_2\cdot(\vec{k}_1\times\vec{k}_2)(\vec{k}_1\cdot\vec{k}_2)\\
\nonumber&&\times\Big[-i\times i^5\int_{-\infty}^0 d\tau a(\tau)^{-3}G_{\varphi,+}(q_1,\tau)G_{\varphi,+}^\prime(q_2,\tau)G_{\varphi,+}(k_1,\tau)G_{\varphi,+}^\prime(k_2,\tau)\\
\nonumber&&~~~+i\times (+i)^5\int_{-\infty}^0 d\tau a(\tau)^{-3}G_{\varphi,-}(q_1,\tau)G_{\varphi,-}^\prime(q_2,\tau)G_{\varphi,-}(k_1,\tau)G_{\varphi,-}^\prime(k_2,\tau)\Big]\\
\nonumber&=&4\left(\frac{\rho_{1,Z}}{\dot{\phi}_0}\rho_2\right)^2\frac{c H^{12}}{ m_Z^6 m_h^4}\frac{\vec{q}_2\cdot(\vec{k}_1\times\vec{k}_2)(\vec{k}_1\cdot\vec{k}_2)}{2q_1^3 2q_2 2k_1^3 2k_2}\\
\nonumber&&~~~\times 2i \Im \left[\int_{-\infty}^{0}d\tau \tau^5(1+i q_1\tau)(1+i k_1\tau)e^{-iK\tau}\right]\\
&=&0~.
\end{eqnarray}
The above expression vanishes because the time integral gives a real result. Therefore the LO term turns out to be unobservable in the trispectrum in dS. A similar calculation can be performed for the NLO term and we still get a null result because of the vanishing imaginary part of the time integral. These null results can also be viewed as a cancellation between the time-ordered diagram and the anti-time-ordered diagram. The cancellation is irrespective of the assumption $m_h\gg m_Z$ above and should be quite general. In fact, similar cancellation happens to the P-odd shape graviton bispectrum in exact dS spacetime \cite{Maldacena:2011nz,Soda:2011am,Shiraishi:2011st}. 

This phenomenon is more easily understood in the wavefunction formalism. The wavefunction of the universe in single-field EFT is given by
\begin{equation}
	\Psi[\varphi]=\mathcal{N}\exp\left(-\frac{1}{2}\int \psi_2\varphi^2-\frac{1}{4!}\int\psi_4\varphi^4+\cdots\right)~,
\end{equation}
where $\psi_N$'s are the dual correlators of the boundary CFT. In particular, the power spectrum and the trispectrum are given by the relations
\begin{eqnarray}
	\langle\varphi^2\rangle'&=&\frac{1}{2\Re^\prime\tilde{\psi}_2}\\
	\langle\varphi^4\rangle'&=&\frac{-2\Re'\tilde{\psi}_4}{(2\mathrm{Re}^\prime\tilde{\psi}_2)(2\mathrm{Re}^\prime\tilde{\psi}_2)(2\mathrm{Re}^\prime\tilde{\psi}_2)}~,
\end{eqnarray}
where $\tilde{\psi}_N$'s are written in momentum space and $\Re'$ is defined as $\Re'(\#)\equiv\frac{1}{2}(\#+\#^*|_{\vec{p}_n\to -\vec{p}_n})$. For a perturbative calculation, $\tilde{\psi}_4$ is essentially the time-ordered diagram in (\ref{LOnull}). However, if the time integral gives a real value, $\Re'\tilde{\psi}_4=\frac{1}{2}(\tilde{\psi}_4+\tilde{\psi}_4^*|_{\vec{p}_n\to -\vec{p}_n})=\frac{1}{2}(\tilde{\psi}_4-\tilde{\psi}_4)=0$. Thus even if $\tilde{\psi}_4\neq 0$, the trispectrum is still zero. Viewed another way, a P-odd real $\tilde{\psi}_4$ yields a purely imaginary $\psi_4$ in coordinate space, $i.e.$, $\psi_4=\pm i |\psi_4|$. This suggests that the wavefunction is modified by a pure phase due to the four-point interaction,
\begin{equation}
	\Psi[\varphi]=\mathcal{N}\exp\left(-\frac{i}{4!}\int\pm|\psi_4|\varphi^4\right)\exp\left(-\frac{1}{2}\int \psi_2\varphi^2+\cdots\right)~.
\end{equation}
This pure phase is eliminated when computing the expectation value of an observable \cite{Maldacena:2011nz}:
\begin{equation}
	\langle\mathcal{O}\rangle=\frac{\int \mathcal{D}\varphi |\Psi|^2 \mathcal{O}}{\int \mathcal{D}\varphi |\Psi|^2}~.
\end{equation}
As a result, to check whether we can obtain power-law suppressed CP-violating effects in the trispectrum, we only need to check whether $\tilde{\psi}_4$ (in other words, the time integral) computed using (\ref{EFTExpansion}) is real or not. 

The most general P- and CP-odd 4-point contact vertex in the large mass EFT can be schematically written as
\begin{equation}\label{AllEFTop}
	\int d\tau d^3x a^{1-2m-n} \epsilon^{(3)}\partial_i^{3+2m} \partial_\tau^n\varphi^4~,
\end{equation}
where $\epsilon^{(3)}$ is the Levi-Civita symbol for three spatial dimensions and is assumed to be contracted to the spatial derivatives ($m\geqslant0$). The power of the scale factor is fixed by dilation in exact dS spacetime. Also without any loss of generality, we restrict $n\leqslant3$ by using inflaton EOM and IBPs. Therefore, $\tilde{\psi}_4$ is given by
\begin{subequations}
	\begin{eqnarray}
	\tilde{\psi}_4&\propto& i \epsilon^{(3)}(ik_i)^{3+2m}\int_{-\infty}^0 d\tau a^{1-2m-n}\partial_\tau^n G_+^4\\
	&\propto&\epsilon^{(3)}(k_i)^3\times (k_j)^{2m}\int_{-\infty}^0 d\tau \tau^{2m+n-1}\partial_\tau^n u^{*4}\label{timeDepOfScaleFactor}\\
	&\propto&\epsilon^{(3)}(k_i)^3\times (k_j)^{2m}\int_{-\infty}^0 d\tau \tau^{2m+n-1}\partial_\tau^n \left(1-k\frac{\partial}{\partial k}\right)^4 e^{iK\tau}\label{beiguosuanfu}\\
	&\propto&\epsilon^{(3)}(k_i)^3\times (k_j)^{2m}\left(1-k\frac{\partial}{\partial k}\right)^4\int_{-\infty}^0 d\tau \tau^{2m+n-1}(iK)^n e^{iK\tau}\\
	&\propto&\epsilon^{(3)}(k_i)^3\times (k_j)^{2m}\left(1-k\frac{\partial}{\partial k}\right)^4(iK)^n (-1)^{n+1}(iK)^{-2m-n}\Gamma(2m+n)\\
	&\propto&\epsilon^{(3)}(k_i)^3\times (k_j)^{2m}\left(1-k\frac{\partial}{\partial k}\right)^4 K^{-2m}\in\mathbb{R} ~.
	\end{eqnarray}\label{proof4ptc}
\end{subequations}
Here in every step some real factors are omitted for simplicity. In (\ref{beiguosuanfu}) we have evoked the symmetry-breaking operator introduced in \cite{Chu:2018ovy}. The final $\tilde{\psi}_4$ is manifestly real for $2m+n>0$. Yet it is P-odd because of the Levi-Civita symbol. Henceforth, by the preceding argument, none of the EFT operators of the form (\ref{AllEFTop}) contribute to the in-in observables. They are all pure phases of the wavefunction. This null result can be generalized to all tree-level diagrams and we sketch the proof in Appendix~\ref{EFTtreeAppx}.

Several important remarks are given below.

First, the P-odd local EFT operators are unobservable not because of the indistinguishability between the external lines, since cancellation takes place before permutation. In other words, P-odd local EFT operators made by four massless fields with different flavors are also unobservable in the 4-point function in exact dS.

Second, integrating out extra fields usually yields a modification of the inflaton sound speed. However a constant change of sound speed does not influence the conclusion above since it can be absorbed into a global redefinition of spatial coordinates.

Third, dilation symmetry is crucial in this derivation. For example, any small departure of the scale factor in (\ref{timeDepOfScaleFactor}) from the dS case, where $a(\tau)\propto\tau^{-1}$, can easily produce an imaginary part for $\tilde{\psi}_4$. As a special case, when spacetime is flat, $a\equiv 1$, the 4-point function will receive contributions from the P-odd contact EFT operators. Thus in the real inflationary universe, it is possible to have observable P- and CP-violating effects in a local EFT, but they are subjected to a slow-roll suppression (see \cite{Soda:2011am,Shiraishi:2011st} for similar statements in the case of graviton bispectrum).

Indeed, the large-mass fall-off behavior of the P- and CP-violating signal obtained in our model is faster than the naive power laws in (\ref{EFTLOandNLO}). This suggests that this signal is suppressed by exponential factors like $e^{-\pi\mu_h}$ or $e^{-\pi\mu_Z}$, which decrease faster than any powers of $\frac{H}{m_Z}\sim \frac{1}{\mu_Z}$ or $\frac{H}{m_h}\sim \frac{1}{\mu_h}$. Hence it is impossible to see the P- and CP-violating signal at any fixed order in the EFT expansion. The signal is caused by the non-perturbative on-shell particle creation in the dS spacetime. Real particles, when produced, automatically break dilation symmetry, which is crucial in the proof of the cancellation of EFT signals above. As a further confirmation, we show in Appendix~\ref{betaCheck} that once the $\beta_\lambda$ coefficient in (\ref{IRexpansion}) is turned off, the signal decreases dramatically by another exponential factor.

Interestingly, we can recast this statement in a more practical form. During inflation, the P- and CP-violating signals from single-field EFT terms are at least slow-roll suppressed. Henceforth, given the measurement bounds on the slow-roll parameters, perturbative unitarity put a constraint on the signal strength of P- and CP-violating scalar trispectrum. Any observation of these signals that exceeds the maximum amount allowed by perturbative single-field EFT alone is an indication of extra degrees of freedom being excited from the BD vacuum during inflation. This is a characteristic feature of quasi-single field inflation and is essential for utilizing the cosmological collider in the possible future.

\section{Conclusions and outlook}\label{Conclusions}
In this work, we studied the simplest P and CP violating signals on the cosmological collider, which opens up a new window of probing P and CP properties of fundamental physics at very high energy scales. We presented a simple model consisting of a Higgs sector and a P- and CP-violating $\theta$ term as a proof of concept that demonstrates this idea. In this model, the 4-point correlation function of the primordial fluctuations is P-odd and possesses an imaginary part with odd dihedral-angle dependence. 
This is very similar to the 4-body decay of a heavy scalar field $X$ that has been routinely studied in collider physics. We studied the perturbation theory description, the partially non-perturbative EOM method as well as the large-mass behavior of this model. For suitable choices of parameter values, the signal strength can be as large as $\tau_{NL}\sim\mathcal{O}(100)$, which is promising for future measurements. 


We showed that the CP-violating local effective operators with the inflaton field alone are all unobservable in an exact dS background. So the CP violating effects in our model will be exponentially suppressed by masses. They are significant only when these masses are close to or below the Hubble scale, in which case the intermediate particles will be created on-shell. Therefore, our study suggests that the CP violating effects in the imaginary part of the trispectrum is usually accompanied by the conventional cosmological collider signals, in the real part of the collapse-limit trispectrum. 

In real inflationary scenario, the local inflaton EFT description of CP-violations is subject to slow-roll suppression. Thus unitarity bound and slow-roll parameters constraints give a maximum amount of CP-violation on the trispectrum in the EFT description. If signals larger than this bound is observed, it is likely to be caused by non-local particle production effects such as in our model. Therefore finding large P- and CP-violations in the scalar trispectrum would be able to serve as a good omen of cosmological collider physics.

Many questions along this direction are left untouched in the current work which we hope to address in the future. We conclude this paper by mentioning a few of them. First, the model we are considering is still based on the EFT operators and is not UV-complete. It will be interesting to embed it in a UV-complete model of particle physics. Second, in our calculation, fermions do not show up in the external lines, thus information of their CP violations are washed out in the observables on the future boundary of dS, which are directly relevant to the late time universe. Although we can predict the CP violation effects in the fermion sector, it is still unclear how one can observe these effects in the boundary correlators of curvature perturbations. Third, in the numerical examples shown in this work, we have assumed that the intermediate particles are heavy, $i.e.$, $m_Z>H/2$ so that we have oscillatory signals in the real part of the trispectrum. This however is not the only possibility since we can readily generalize the calculation to the case of light fields with $m_Z<H/2$, where the signals are expected to be larger. Fourth, in this work we did not perform a systematic analysis of CP-violating single-field EFT in inflation, where dS isometries are softly broken. Thus we did not give any concrete bound on CP-violation in the single-field EFT description. It is important to have an explicit demonstration of this idea in the future work. Fifth, it is interesting to consider the relation between our model and the spontaneous baryogenesis scenario \cite{Cohen:1987vi,Cohen:1988kt,Turner:1988sq}, where net baryon number is created from the relaxation of a pseudo-Goldstone boson field of broken $U(1)_B$. If the $\theta$ term is embedded into UV models such as this, its chemical potential on baryons may lead to spontaneous baryogenesis. The relevant cosmological collider signals may give possible hints for spontaneous baryogenesis.

\begin{acknowledgments} 
We would like to thank Andrew Cohen, Kunfeng Lyu and Siyi Zhou for helpful discussions and comments. 
\end{acknowledgments}

\appendix
\section{Feynman rules for Schwinger-Keldysh diagrammatics}\label{FeynmanRulesAppendix}
The Schwinger-Keldysh propagators are obtained from 2-point functions with different time-orderings. Two of them carrying signature $(+-)$ and $(-+)$ are disconnected contributions due to the sewing condition on two time contours at measurement time. For scalar fields, they are simply
\begin{subequations}
	\begin{eqnarray}
		G_{I;+-}(k,\tau_1,\tau_2)&=&v_{k}^I(\tau_1)^*v_{k}^I(\tau_2)\\
		G_{I;-+}(k,\tau_1,\tau_2)&=&v_{k}^I(\tau_1)v_{k}^I(\tau_2)^*~.
	\end{eqnarray}
\end{subequations}
Here $I=\phi,h,\varphi,\cdots$ stands for different scalar fields\footnote{Notice that here $I$ is not summed over.}, with $v_{k}^I$ being the corresponding mode function,
\begin{equation}
	v_{k}^I(\tau)=-i\frac{\sqrt{\pi}}{2}e^{i\pi\left(\frac{i\mu_I}{2}+\frac{1}{4}\right)}H(-\tau)^{3/2}H^{(1)}_{i\mu_I}(-k\tau),\quad \mu_I = \sqrt{\frac{m_I^2}{H^2}-\frac{9}{4}}~.
\end{equation}
The other two carrying signature $(++)$ and $(--)$ are Green functions of the linear EOMs and represent the propagation of a wave mode.
\begin{subequations}
	\begin{eqnarray}
	G_{I;++}(k,\tau_1,\tau_2)&=&\Theta(\tau_1-\tau_2)G_{I;-+}(k,\tau_1,\tau_2)+\Theta(\tau_2-\tau_1)G_{I;+-}(k,\tau_1,\tau_2)\\
	G_{I;--}(k,\tau_1,\tau_2)&=&\Theta(\tau_1-\tau_2)G_{I;+-}(k,\tau_1,\tau_2)+\Theta(\tau_2-\tau_1)G_{I;-+}(k,\tau_1,\tau_2)~.
	\end{eqnarray}
\end{subequations}
Diagrammatically, $(+)$ sign is represented by a black dot while $(-)$ is represented by a white dot. When one end of the propagator is on the boundary $\tau=0$, these four propagators reduce to two varieties, namely, $G_{I;+}(k,\tau)=G_{I;-+}(k,0,\tau)=G_{I;++}(k,0,\tau)$ and $G_{I;-}(k,\tau)=G_{I;+-}(k,0,\tau)=G_{I;--}(k,0,\tau)$. The boundary end is usually represented by a white square.

The vertex rules for computing correlation functions of the Maxwell-Chern-Simons theory in dS spacetime are given in (\ref{MCSindSFeynRules}), and those for the more realistic model (\ref{AZHIEFTop}) are given in (\ref{AZHIFeynRules}). Each vertex has two choices of color, accounting for their time-ordering. The Feynman rules for negative color vertices are obtained by a complex conjugation and a momentum reversal.
\begin{subequations}\label{MCSindSFeynRules}
	\begin{eqnarray}
	\begin{gathered}
	\includegraphics[width=3cm]{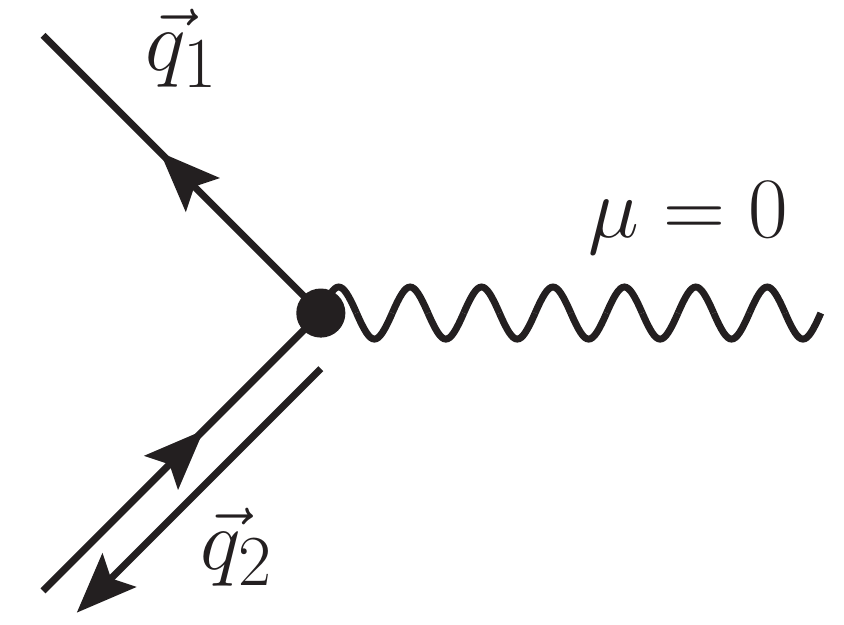}
	\end{gathered}&=&e\int d\tau a(\tau)^2(G_1 G_2^{*\prime}-G_1^{\prime}G_2^{*})D_0\\
	\begin{gathered}
	\includegraphics[width=3cm]{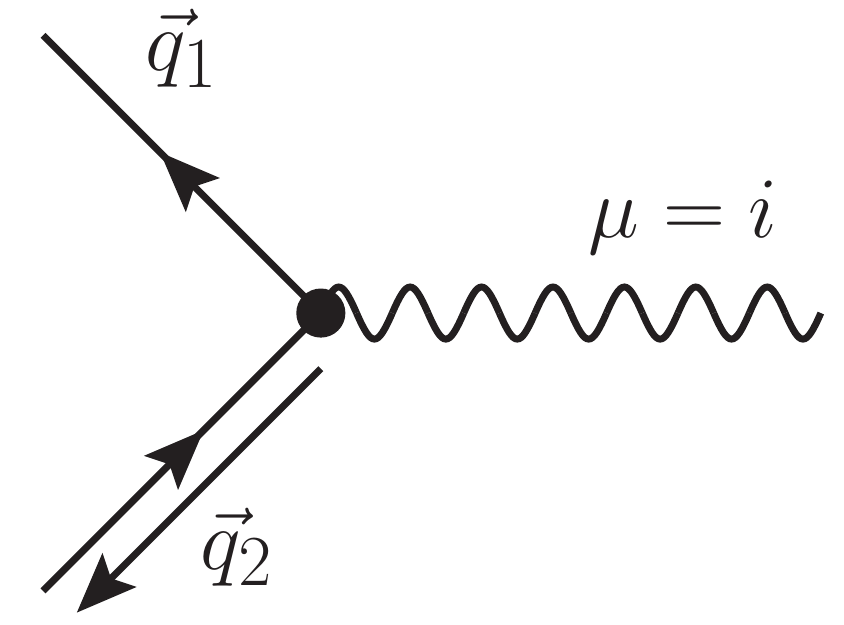}
	\end{gathered}&=&-ie\int d\tau a(\tau)^2 G_1 G_2^*(\vec{q}_1-\vec{q}_2)_i D_i\\
	\begin{gathered}
	\includegraphics[width=3cm]{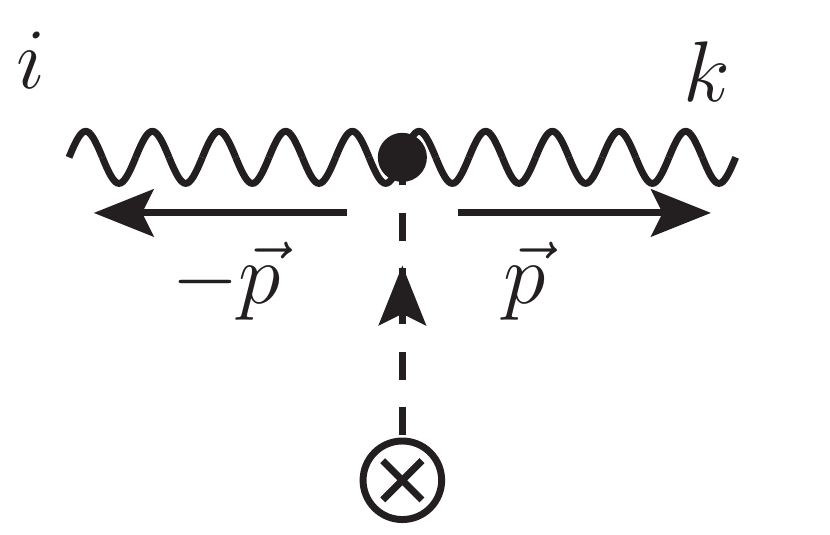}
	\end{gathered}&=&c_0\int d\tau \theta'(\tau) \epsilon^{ijk}D_i p_j D_k~.
	\end{eqnarray}
\end{subequations}

\begin{subequations}\label{AZHIFeynRules}
	\begin{eqnarray}
	\begin{gathered}
	\includegraphics[width=3cm]{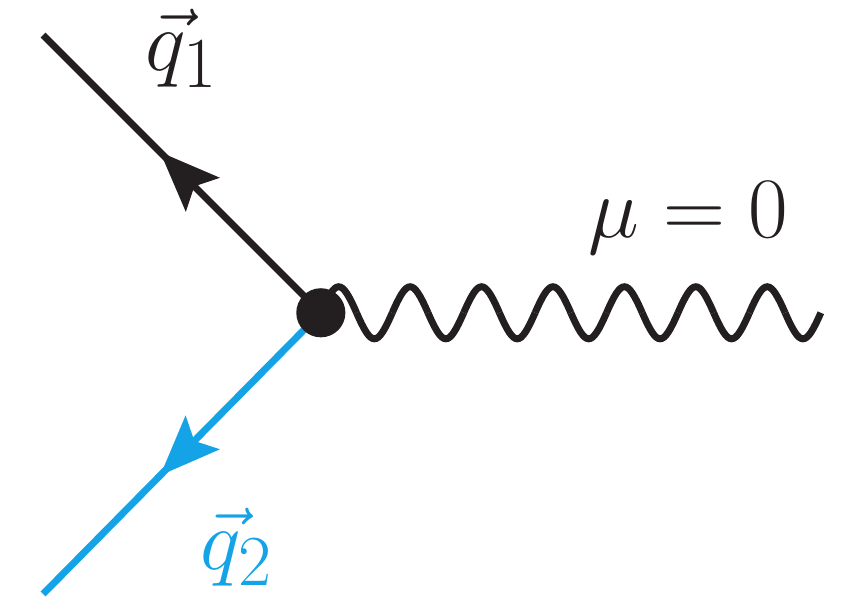}
	\end{gathered}&=&-i\frac{\rho_{1,Z}}{\dot{\phi}_0}\int d\tau a(\tau)^2G_h G_\varphi' D_0\\
	\begin{gathered}
	\includegraphics[width=3cm]{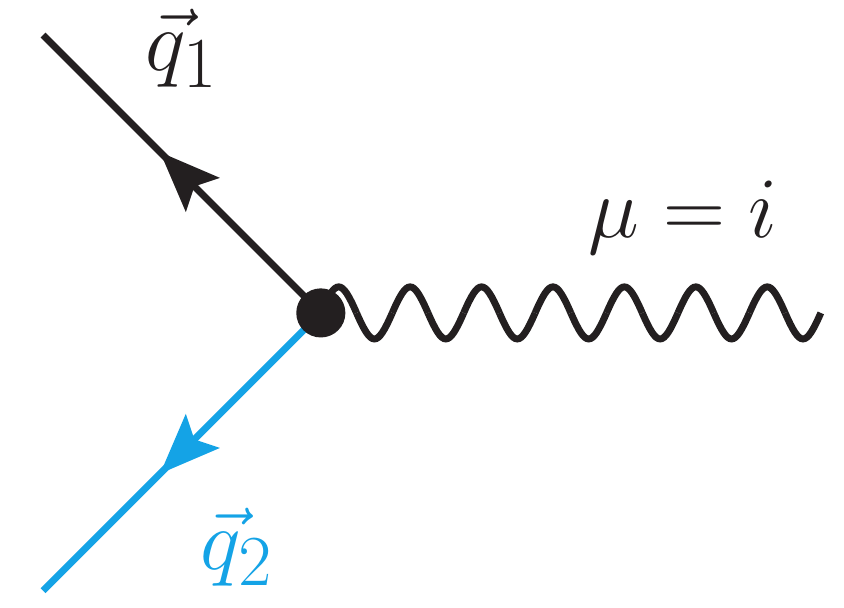}
	\end{gathered}&=&\frac{\rho_{1,Z}}{\dot{\phi}_0}\int d\tau a(\tau)^2 G_h G_\varphi q_{1i} D_i\\
	\begin{gathered}
	\includegraphics[width=3cm]{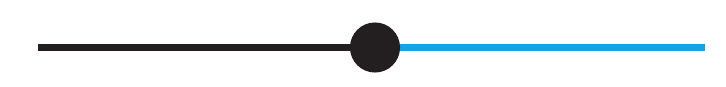}
	\end{gathered}&=&-i\rho_2\int d\tau a(\tau)^3 G_h G_\varphi'\\
	\begin{gathered}
	\includegraphics[width=3cm]{thetaZZVertex}
	\end{gathered}&=&c_0\int d\tau \theta'(\tau) \epsilon^{ijk}D_i p_j D_k~.
	\end{eqnarray}
\end{subequations}

\section{Tree-level diagrams in single-field EFT in dS}\label{EFTtreeAppx}

In this appendix we generalize the discussions in Sect.~\ref{LargeMassEFT} for CP-violating single-field EFT in dS. We will prove that in an exact dS background, any single-field EFT operator of the form
\begin{equation}\label{CP-oddEFTop}
	\Delta \mathcal{L}_{N}^{\text{CP-odd}} = a^{1-2m-n} \epsilon^{(3)}\partial_i^{3+2m} \partial_\tau^n\varphi^N~
\end{equation}
does not lead to observable CP-violating signals at tree-level.

\subsection{A theorem for CP-even EFTs}
In order to show the absence of CP-violation signals, we work in the EFT without P- and CP-violation first and prove a theorem that is by itself meaningful.

\begin{theorem}
	The tree-level contribution to the wavefunction exponent $\tilde{\psi}$ is real for massless scalar EFTs respecting parity, dilation and rotational invariance.
\end{theorem}

\begin{proof}
	Consider a general connected tree diagram with $E$ external lines, $I$ internal lines and $V$ vertices. The set of external lines is denoted $\mathcal{E}$ while the set of internal lines is denoted $\mathcal{I}$. The set of vertices is denoted $\mathcal{V}$. Then a perturbative calculation of $\tilde{\psi}_E$ is given schematically by
	\begin{equation}
	\tilde{\psi}_E=\int_{-\infty(1-i\epsilon)}^{0}\prod_{v\in \mathcal{V}}i d\tau_v D_{v}\prod_{e\in\mathcal{E}}K_e\prod_{e'\in\mathcal{I}}G_{e'}~,
	\end{equation}
	where $D_{v}$ is the differential operator at vertex $v$ originating from the interaction $\Delta\mathcal{L}_{v,N_v}=D_v\varphi^{N_v}$ and is understood to contract with the $N_v$ propagators in all possible ways. The bulk-to-boundary propagator and bulk-to-bulk propagator for a massless scalar are given by \cite{Anninos:2014lwa}:
	\begin{eqnarray}\label{wavefunProps}
	K(k,\tau)&=&(1-i k \tau)e^{i k\tau}\\
	\nonumber G(k,\tau_1,\tau_2)&=&\frac{H^2}{2k^3}(1+ik\tau_1)(1-ik\tau_2)e^{-ik(\tau_1-\tau_2)}\Theta(\tau_1-\tau_2)+(\tau_1\leftrightarrow\tau_2)\\
	&&-\frac{H^2}{2k^3}(1-ik\tau_1)(1-ik\tau_2)e^{ik(\tau_1+\tau_2)}~.
	\end{eqnarray}
	The bulk-to-boundary propagator solves homogeneous EOM while the bulk-to-bulk propagator solves EOM with a point source under a Dirichlet boundary condition. For convenience, we change to the variable $x\equiv-\tau$. Then $\tilde{\psi}_E$ takes the form
	\begin{equation}
	\tilde{\psi}_E=\int_{0}^{\infty(1-i\epsilon)}\left[\prod_{v\in \mathcal{V}}i dx_v D_v\prod_{e\in\mathcal{E}}K_e\prod_{e'\in\mathcal{I}}G_{e'}\right]_{\tau_v\rightarrow-x_v}~.
	\end{equation}
	
	The differential operators representing local interactions act on the time-ordering $\Theta$-function part as well as the mode-function part of the propagators. The case of a single time derivative acting on the propagator contains no contribution from differentiating the $\Theta$-function due to the symmetric property $G(k,\tau_1,\tau_2)=G(k,\tau_2,\tau_1)$. However, there arises a new term from differentiating the $\Theta$-function when two time derivatives are involved. In terms of the massless mode function $v_k(\tau)=\frac{H}{\sqrt{2k^3}}(1+ik\tau)e^{-ik\tau}$, a straightforward computation gives
	\begin{eqnarray}
		\nonumber\partial_{\tau_1}\partial_{\tau_2}G(k,\tau_1,\tau_2)&=&\Theta(\tau_1-\tau_2)v_k^\prime(\tau_1)v_k^\prime(\tau_2)^*+\Theta(\tau_2-\tau_1)v_k^\prime(\tau_1)^*v_k^\prime(\tau_2)-v_k^\prime(\tau_1)^*v_k^\prime(\tau_2)^*\\
		&&+\delta(\tau_1-\tau_2)\left(v_k v_k^{\prime*}-v_k^\prime v_k^*\right)(\tau_1)~,
	\end{eqnarray}
	where the last term is a new contact contribution proportional to the Wronskian of the mode function, $v_k v_k^{\prime*}-v_k^\prime v_k^*=i a^{-2}$. Thus a contact interaction is induced from joining two derivative interactions with a propagator. In the Hamiltonian approach, this is essentially due to derivative interactions affecting the definition of canonical momenta \cite{Chen:2017ryl}. Being a tree-level effect, the induced contact interaction apparently still preserves parity, dilation and rotational symmetry and so is still within the tower of EFT operators. And every time a contact interaction is induced, the original diagram becomes a diagram with one internal line short (see FIG.~\ref{ReductionExample}). Hence in this way we can reduce a given tree diagram with $E$ external lines to a set of irreducible diagrams, each with $E$ external lines and without contact contributions hidden in the propagators. This set of diagrams is the same as the set of tree diagrams generated by the tower of EFT operators, except for (diagram-dependent) real numeric factors which is not important for our purpose.
	\begin{figure}[h!]
		\centering
		\includegraphics[width=15cm]{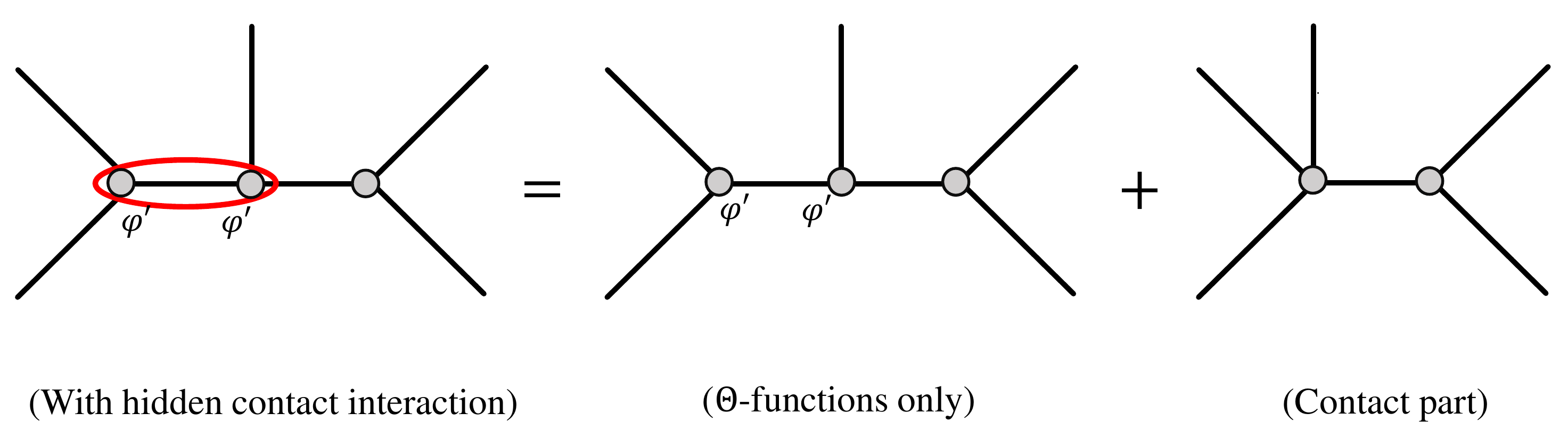}
		\caption{An illustration of the reduction process. Here we take $E=5$ and the red circle indicates that there is a hidden contact contribution in the derivative interactions. We can separate this contribution and form the second diagram on the right, with the first diagram containing $\Theta$-functions only. But the second diagram already exists in the perturbation theory, so this merely modifies its coefficient by a real numeric factor.}\label{ReductionExample}
	\end{figure}

	Now having reduced the hidden contact interactions, we have obtained diagrams in which only $\Theta$-functions are present. Choose any such irreducible diagram with $I$ internal lines and $V$ vertices, we can first expand it to the sum of $3^I$ terms, each term having a particular integration ordering. These different integration orderings define various subregions of $\mathbb{R}_+^V$. After further expanding the integrand to a sum of monomials, we see that they all take the same form
	\begin{equation}\label{monoimialForm}
	(\text{Real factor})\times\int_\mathbb{T}i^V d^Vx \prod_{j=1}^V i^{A_j} x_j^{B_j} e^{-i\sum_j p_j x_j}~,
	\end{equation}
	where $\mathbb{T}\subset \mathbb{R}_+^V$ is the integration region and $p_j$'s are linear combinations of internal and external momenta, which will be of no importance in our proof. The relevant quantities here for the reality of (\ref{monoimialForm}) are the powers $A_j$ and $B_j$, as can be seen after a Wick rotation. According to the $i\epsilon$-prescription, the integral factorizes into
	\begin{equation}
	i^{\sum_j(A_j+B_j)}\times(\text{Real factor})\times(-1)^{\sum_j B_j}\int_\mathbb{T}d^Vx \prod_{j=1}^V x_j^{B_j} e^{-\sum_j p_j x_j}~.
	\end{equation}
	The integral is now purely real and the only (possibly) imaginary part comes from the prefactor, which is real if $\sum_j(A_j+B_j)$ is even.
	
	By a simple power-counting method, we can show that $A_j+B_j$ is even for any $j$.
	
	First let us consider interactions without time derivatives. Parity, dilation as well as rotational invariance fix the general form of the interaction to be
	\begin{equation}
	D_j=a^{4-2m}\partial_i^{2m}~.
	\end{equation}
	Since we are working in the perfect dS limit, each scale factor $a(\tau_j)\propto \tau_j^{-1}\propto x_j^{-1}$ decrease $B_j$ by $-1$. Thus the scale factors in the interaction vertex change $B_j$ by $\Delta B_j=2m-4$. From (\ref{wavefunProps}) we see that $i$ and $\tau_j=-x_j$ always appear together in the propagators, adding $\Delta A_j=\Delta B_j=1$ every time expanded. As a result, starting with zero, $A_j+B_j$ always changes by an even number, giving rise to an even $A_j+B_j$ in the end.
	
	Second, consider interactions with time derivatives,
	\begin{equation}
	D_j=a^{4-2m-n}\partial_i^{2m}\partial_{\tau_j}^n~.
	\end{equation}
	Each time derivative acting on (\ref{wavefunProps}) either brings down a $-i$ from the exponential or lowers the power of $\tau_j$ by one\footnote{Note that since we have performed the reduction of hidden contact interactions, there is no time derivatives acting on $\Theta$-functions now.}, leading to a net change $\Delta A_j+\Delta B_j=-1$. This is canceled by an extra power of $a^{-1}$ in the interaction vertex. Therefore the total $A_j+B_j$ is again even.
	
	In summary, the prefactor $i^{\sum_j(A_j+B_j)}\in \mathbb{R}$ and every monomial (\ref{monoimialForm}) in the expansion of $\tilde{\psi}_E$ is real, leading to the conclusion $\tilde{\psi}_E\in \mathbb{R}$.
	
\end{proof}

We point out that our proof is not dependent on the flavor of the scalar field. Hence the theorem also applies to the case with multiple massless scalar fields.

\subsection{The absence of CP-violation signals}
With the help of the above theorem, we can now proceed to consider the CP-violating case.

First, it is easy to show that a single $N$-point contact diagram generated by $\Delta \mathcal{L}_{N}^{\text{CP-odd}}$ in (\ref{CP-oddEFTop}) is unobservable since the corresponding $\tilde{\psi}_N\in\mathbb{R}$ and is P-odd. Its proof is identical to (\ref{proof4ptc}), with powers of 4 replaced by $N$. Therefore, any \textit{disconnected} tree diagram made of this piece of $\tilde{\psi}_N\in\mathbb{R}$ vanishes for observables because it is a pure phase before the wavefunction.

Second, we can show that CP-violating signals in \textit{connected} tree diagrams are also unobservable. For any EFT vertex $\Delta \mathcal{L}_{N}^{\text{CP-odd}}$, we can relate it to a CP-even interaction vertex obtained by removing the $\epsilon^{ijk}\partial_i\partial_j\partial_k$ structure and adding three scale factors,
\begin{equation}
\Delta \mathcal{L}_{N}^{\text{CP-even}} = a^{4-2m-n} \partial_i^{2m} \partial_\tau^n\varphi^N~.
\end{equation}
Now, for each connected tree diagram containing a vertex $\Delta \mathcal{L}_{N}^{\text{CP-odd}}$, there is a corresponding diagram with the replacement $\Delta \mathcal{L}_{N}^{\text{CP-odd}}\rightarrow \Delta \mathcal{L}_{N}^{\text{CP-even}}$ (see FIG.~\ref{OldVsNew}). If we count the power of the scale factor, there is an extra phase $i^3$ coming from time integral, which cancels the $i$ coming from three spatial derivatives. Thus the new diagram has a same phase (up to $\pi$) as the old diagram.
\begin{figure}[h!]
	\centering
	\includegraphics[width=11cm]{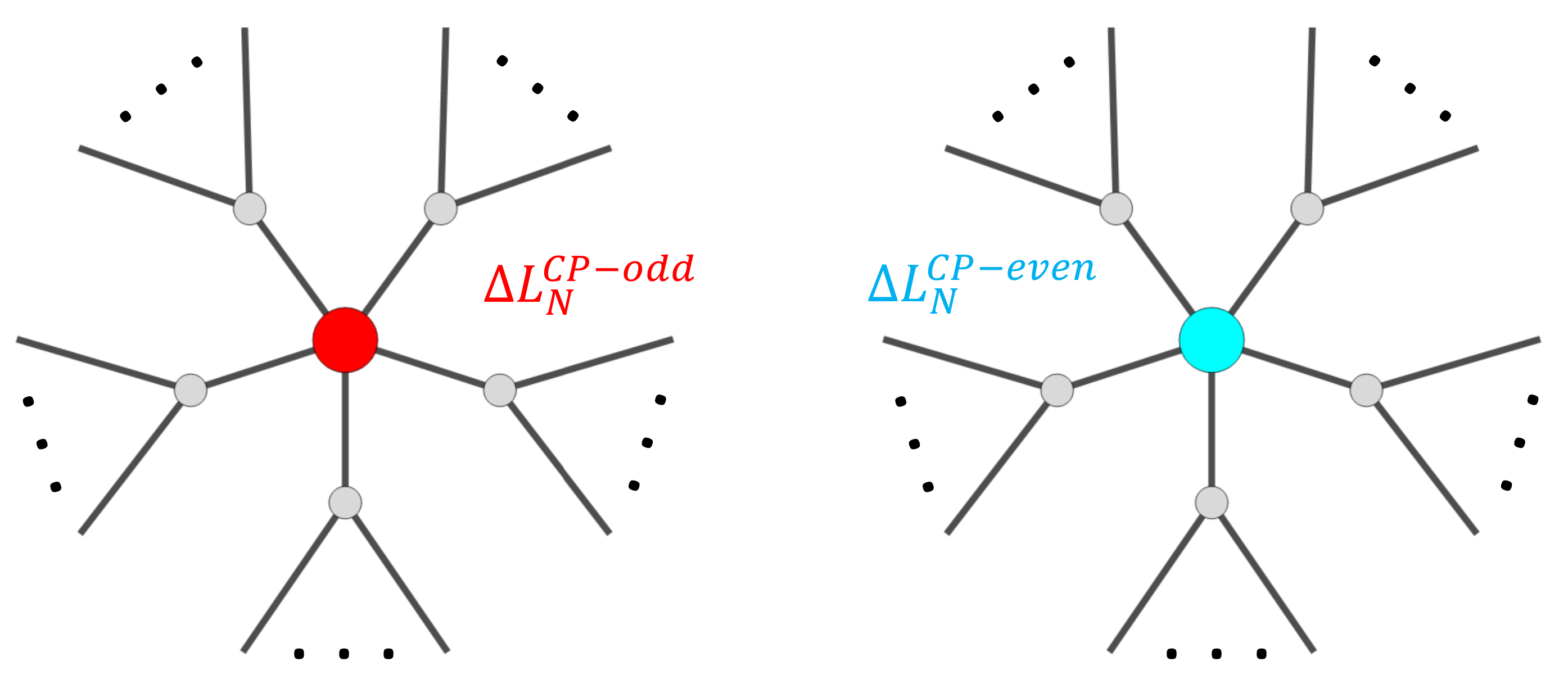}
	\caption{The old diagram (left) and the new diagram (right) differ by three spatial derivatives and three powers of scale factors, thus they share the same phase.}\label{OldVsNew}
\end{figure}

Since the new diagram comes from a massless EFT \textit{without} P-violation, according to the theorem proven in the proceeding subsection, it must be real. Consequently, the old diagram is also real\footnote{An aside: One can also directly show this by the power-counting method used in the proof of the previous theorem.}. In addition, it flips sign under a spatial reflection. Then by the pure-phase argument, the old \textit{connected} tree diagram also does not contribute to the observables.

The above $\Delta \mathcal{L}_{N}^{\text{CP-odd}}\rightarrow \Delta \mathcal{L}_{N}^{\text{CP-even}}$ treatment immediately generalizes to an odd number of $\Delta \mathcal{L}_{N}^{\text{CP-odd}}$ insertions. In contrast, an even number of $\Delta \mathcal{L}_{N}^{\text{CP-odd}}$ insertions does contribute to observable correlations functions, yet in that case CP is preserved. Thus, in conclusion, we have shown that the CP-violation signals are indeed absent in the tree-level observables.


\section{A check on the mass dependence}\label{betaCheck}
Here we numerically check the dependence of the strength of the CP-violating signal on the mass of the gauge boson, keeping the mass of the Higgs fixed. The result is plotted in FIG.~\ref{ExpCheck}. To show the importance of particle creation, we also use the IR expansion (\ref{IRexpansion}) and turn off the negative-frequency modes by setting $\beta_\lambda\equiv 0$, which eliminates the particle creation effects up to higher orders in $e^{-\pi\mu_Z}$. As is clear from FIG.~\ref{ExpCheck}, the CP-violating signal decreases approximately as $e^{-\pi\mu_Z}$. In addition, the IR expansion ($\alpha+\beta$) provides a good approximation to the full solution, while the positive-frequency-only result ($\alpha$) is smaller than the full result by yet another exponential factor. This demonstrates the importance of the $\beta$ term.

We should point out that the exponential dependence on $\mu_Z$ is only valid for large $\mu_h$. In other words, it is possible that the CP-violating signal is suppressed only by powers of $1/\mu_Z$, but with a factor $e^{-\pi \mu_h}$ in the front. This corresponds to a non-local $(\Box-m_h^2)^{-1}$ interaction with a local EFT expansion of $(\Box-m_Z^2)^{-1}$. As we have shown in Sect.~\ref{LargeMassEFT}, what is forbidden as an observable is a completely localized CP-violating operator made of four scalars that respects dilation symmetry. 

\begin{figure}[h!]
	\centering
	\includegraphics[width=10cm]{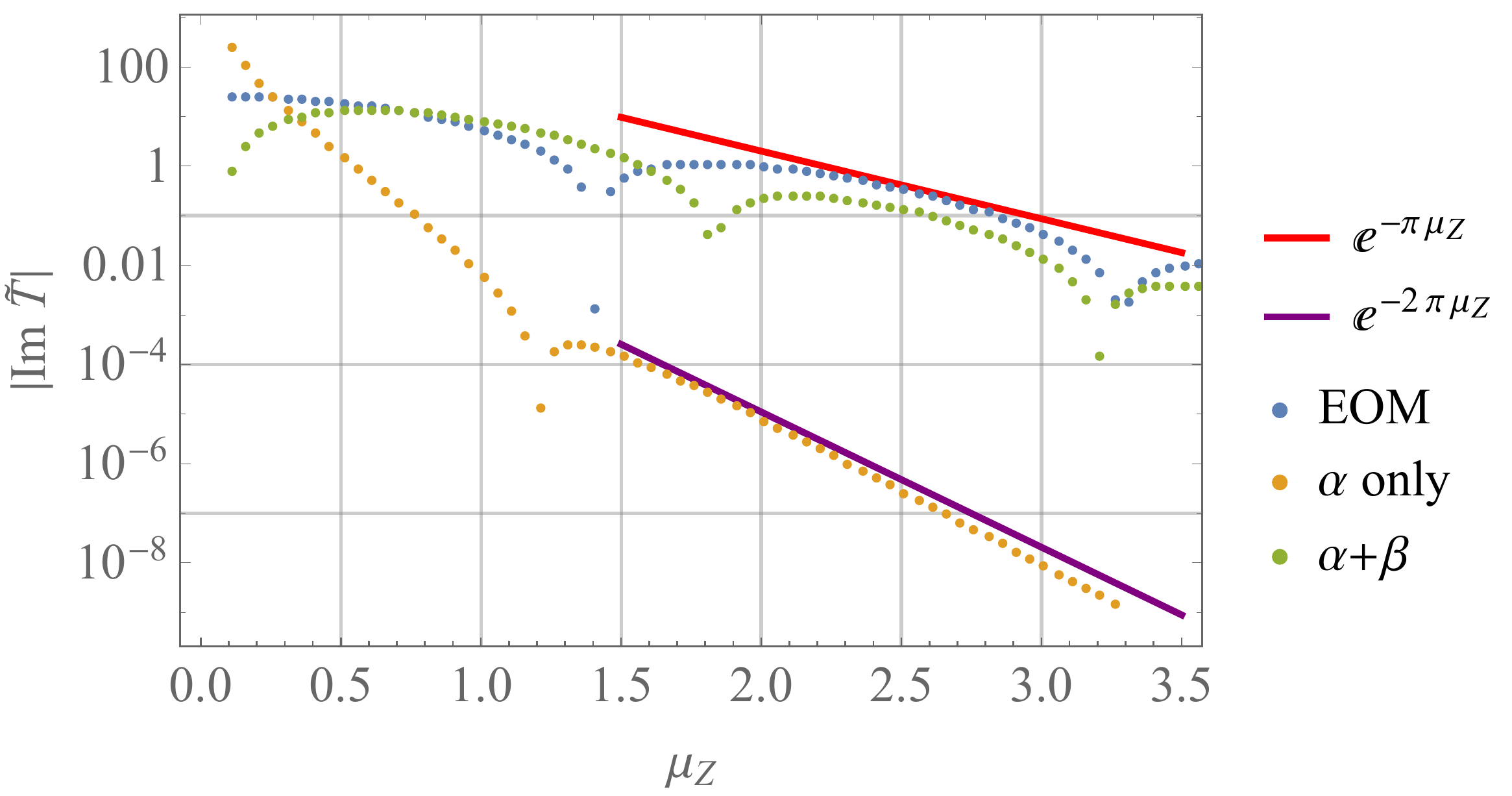}
	\caption{The imaginary part of the dimensionless trispectrum divided by couplings $\tilde{T}\equiv T/(\frac{\rho_2}{H})^2(\frac{\rho_{1,Z}}{\dot{\phi}_0})^2$ as a function of $\mu_Z$. The momentum configuration is chosen to be $k_1=q_1=1,p=2,\theta_k=\theta_q=\pi/3,\phi=0.5$. Here we have taken $c=0.1$ and $\mu_h=2.5$, which correspond to $m_h=2.92H$. The blue dots are from the partially non-perturbative EOM method discussed in Sect.~\ref{EOMmethod}, while the yellow and green dots represent respectively, the $\alpha$-term-only case and case with the full IR expansion containing both $\alpha$ term and $\beta$ term. The red and purple lines show the asymptotic behavior at large $\mu_Z$.}\label{ExpCheck}
\end{figure}

\end{document}